%% file: main.tex
\documentclass[letterpaper,11pt]{article}
\usepackage{graphicx}
\usepackage{amsfonts,amssymb,amsthm,amsmath}
\usepackage[margin=1in]{geometry}
\usepackage{tfp-article}
\hypersetup{pdftitle={The Role of Degrees and Acyclic Neighborhoods in the Complexity of Tournament Fixing},pdfauthor={Yuxi Liu; Junqiang Peng; Zhonghao Wang; Mingyu Xiao}}
\title{The Role of Degrees and Acyclic Neighborhoods in the Complexity of Tournament Fixing}
\author{Yuxi Liu, Junqiang Peng, Zhonghao Wang, Mingyu Xiao\textsuperscript{*}\\[4pt]
\small University of Electronic Science and Technology of China, Chengdu, China\\
\small \textsuperscript{*}Corresponding author\\[2pt]
\small \href{mailto:yuxiliu823@gmail.com}{yuxiliu823@gmail.com}, \href{mailto:jqpeng0@foxmail.com}{jqpeng0@foxmail.com}\\
\small \href{mailto:zhonghao.w@outlook.com}{zhonghao.w@outlook.com}, \href{mailto:myxiao@uestc.edu.cn}{myxiao@uestc.edu.cn}}
\date{}
\begin{document}
\maketitle
\begin{abstract}
The \textsc{Tournament Fixing} problem asks whether a knockout tournament can be scheduled to make a favorite player win when all pairwise match outcomes are given.
Such tournaments are used in sports and provide models for elections and organizational decision-making.
Although the problem is NP-hard in general, it is polynomial-time solvable when some structural parameters become zero.
Motivated by these cases, we study its parameterized complexity using those structural parameters centered on the favorite player.
We prove fixed-parameter tractability when parameterized by either the number of players who beat the favorite player or the number that the favorite player beats.
We next consider four subset parameters that measure the arc or vertex deletions needed to eliminate cycles through the player or either neighborhood.
For each parameter, the problem remains NP-hard even when its value is one.
When both neighborhoods are acyclic, we show that the problem remains NP-hard.
 We also give some sufficient conditions for the existence of a winning seeding.
Our results complete the tractability classification of several individual parameters centered on the favorite player, extending the parameterized complexity landscape of this problem.

\end{abstract}

{\small\noindent\textbf{Funding:} The work is supported by the National Natural Science Foundation of China (Grant No.~62372095).\par}

\section{Introduction}
\input{Sources/ch_IJCAI_introduction.tex}

\section{Preliminaries}\label{sec:preliminaries}
\input{Sources/ch_combined_Preliminaries.tex}

\section{Parameterized by the Out-degree and In-degree}\label{sec:degree-parameters}
\input{Sources/ch_Parameterized_by_the_Out-degree_and_In-degree.tex}

\section{Parameterized by the Subset FAS Number}\label{sec:subset-parameters}
\input{Sources/ch_Parameterized_by_the_Subset_FAS_Number.tex}

\section{Acyclic Neighborhood Structures}\label{sec:acyclic-neighborhoods}   
\input{Sources/ch_Acyclic_Neighborhood_Structures.tex}

\section{Conclusion}\label{sec:conclusion}

We have expanded the parameterized complexity landscape of TFP by studying degree parameters and subset feedback set parameters centered on the favorite player.
Together with the known results for the fas and fvs numbers, our results complete the $\FPT$/para-$\NP$-hard classification of the several individual parameters in Figure~\ref{fig:diagram}.
TFP is $\FPT$ when parameterized by either the in-degree or the out-degree.
All four subset parameters yield $\NP$-hardness even at value $1$.

TFP remains $\NP$-complete when both neighborhoods of the favorite player induce acyclic subtournaments.
Under this restriction, bounding the sfas-v number yields an $\FPT$ algorithm.
We further give sufficient degree conditions that guarantee a winning seeding even when the relevant parameters are unbounded.
These results clarify how local restrictions and their combinations affect the tractability of TFP.

A direction for future work is to improve the parameter dependence of our in-degree algorithm and the algorithms for tractable subset parameterizations.
Further study of the arcs between the two neighborhoods may identify additional classes of efficiently solvable instances.

\ACKNOWLEDGMENT{%

Preliminary versions of this work were presented at the AAAI~\cite{wang2026hard} and IJCAI~\cite{liu2026complexitytournamentfixingsubset} conferences. Compared with the conference versions, there are three major differences. We give a revised proof of Theorem~\ref{thm:acyclic-neighborhood-strict-degree} to avoid a possible incomplete argument cited from literature.
We supply complete proofs omitted from the IJCAI proceedings version.
Furthermore we resolve the open problem posed in that version by proving that TFP remains $\NP$-complete on \DAGins{}s (Theorem~\ref{thm:acyclic-neighborhood-hardness}).
\par
We used OpenAI Codex to assist in proofreading, language polishing, and checking the clarity and logical consistency of the proofs. The authors verified the correctness and originality of all content including references.%
}

\bibliographystyle{informs2014}
\bibliography{aaai2026}
\end{document}

%% file: Sources/ch_IJCAI_introduction.tex
\newcommand{\shortcite}{\citeyearpar}

Knockout tournaments represent a fundamental class of competition structures, characterized by a sequence of recursive elimination rounds~\cite{horen1985comparing,connolly2011tournament,groh2012optimal}. 
In each round, players are partitioned into disjoint pairs (matches). 
For each pair, the loser is eliminated and the winner advances to the next round. 
The process repeats until a single champion remains. 
Owing to their efficiency and simplicity, these structures are employed in large-scale events such as the FIFA World Cup~\cite{scarf2011numerical} and the NCAA Basketball Tournament~\cite{kvam2006logistic}. 
Beyond sports they serve as critical models for elections, organizational decision-making and various game-based settings where sequential elimination serves as a natural selection mechanism~\cite{kim2017can,stanton2011manipulating,ramanujan2017rigging}. 
Consequently knockout tournaments have attracted considerable attention in artificial intelligence~\cite{vu2009complexity,williams2010fixing}, economics and operations research~\cite{rosen1985prizes,mitchell1983toward,laslier1997tournament} (motivating a rich body of theoretical and applied studies).

Given a favorite player $v^*$, a fundamental computational challenge is to determine whether the tournament can be arranged to ensure that $v^*$ wins.
Let $N$ denote the set of $n$ players, where $n = 2^c$ (for some $c \in \mathbb{N}$). 
The structure of a knockout tournament is represented by a complete (unordered) binary tree $T$ with $n$ leaves. 
A \emph{seeding} is a bijection $\sigma : N \to \text{leaves}(T)$ that assigns each player to a distinct leaf. 
The tournament proceeds in rounds: in each round at any internal vertex the winners of the subtrees rooted at its children compete, with the winner advancing to this vertex. 
This continues until one player remains --- the tournament champion.

The question becomes meaningful when predictive information about match outcomes is available. 
We model this with a tournament digraph $D = (V = N, E)$, where $V$ is the set of $n$ players.  
For every distinct pair of vertices $\{u, v\} \subseteq V$ the arc set $E$ contains exactly one of the ordered pairs $(u, v)$ or $(v, u)$, signifying that $u$ beats $v$ or vice versa. 
We investigate the following problem:

\begin{blackbox}
    \textsc{Tournament Fixing Problem} (TFP)

    \noindent\textbf{Input: } A tournament $D$ and a player $v^* \in V(D)$.

    \noindent\textbf{Question: } Does there exist a seeding $\sigma$ for the $n = |V(D)|$ players such that $v^*$ wins the resulting knockout tournament?
\end{blackbox}

\subsection{Related Work}
The computational study of tournament manipulation was initiated by Vu et al.~\shortcite{vu2009complexity}. The complexity of TFP (specifically its $\NP$-hardness) remained open for several years~\cite{vu2009complexity,williams2010fixing,russell2011empirical,lang2012winner}. Aziz et al.~\shortcite{aziz2014fixing} finally proved that TFP is $\NP$-complete. They also provided exact exponential-time algorithms running in $\mathcal{O}(2.83^n)$ time with exponential space, or $4^{n + o(n)}$ time with polynomial space. Kim and Vassilevska Williams~\shortcite{kim2015fixing} later improved this to $2^n n^{\mathcal{O}(1)}$ time and space, and Gupta et al.~\shortcite{gupta2018winning} achieved the same time bound using only polynomial space.

Given the intractability of TFP in general, its parameterized complexity has been extensively investigated. We recall several key parameters studied in the literature and this paper:
\begin{itemize}
    \item \textbf{fas number} (resp. \textbf{fvs number}): the minimum number of arc deletions (resp. vertex deletions) needed to make the tournament acyclic.
    \item \textbf{in-degree} (resp. \textbf{out-degree}) of $v^*$: the number of in-neighbors (resp. out-neighbors) of $v^*$.
    \item \textbf{sfas-v number} (resp. \textbf{sfvs-v number}): the minimum number of arc deletions (resp. vertex deletions, excluding $v^*$) needed to remove $v^*$ from all cycles.
    \item \textbf{sfas-in number} (resp. \textbf{sfas-out number}): the minimum number of arc deletions required to exclude the in-neighborhood of $v^*$ ($\nin(v^*)$) (resp. the out-neighborhood of $v^*$ ($\nout(v^*)$)) from cycles.
\end{itemize}

These parameters share the property that TFP becomes polynomial-time solvable when they equal zero. Their relationships are illustrated in Figure~\ref{fig:diagram} (an arrow from $x$ to $y$ indicates $x \leq y$).
Table~\ref{tab:combined-structural-results} summarizes the results for combined parameters and acyclic neighborhoods.

\input{Figures/fig_survey}
\input{Figures/tab_combined_structural_results}

Ramanujan and Szeider~\shortcite{ramanujan2017rigging} first showed that TFP is fixed-parameter tractable ($\FPT$) when parameterized by the fas number $p$. Their algorithm runs in $p^{\mathcal{O}(p^2)} n^{\mathcal{O}(1)}$ time. This was subsequently improved to $p^{\mathcal{O}(p)} n^{\mathcal{O}(1)}$ by Gupta et al.~\shortcite{gupta2018rigging}.
Later, Gupta et al.~\shortcite{gupta2019succinct} gave a polynomial kernel. For the fvs number $q$, Zehavi~\shortcite{zehavi2023tournament} proved fixed-parameter tractability via a $q^{\mathcal{O}(q)} n^{\mathcal{O}(1)}$-time algorithm. This result subsumes the fas-number parameterization since $q \leq p$.

\subsection{The Motivation and Parameters}
The fas and fvs numbers are structural parameters for TFP that have been extensively studied because the problem becomes trivial when either is zero (i.e. when the tournament is acyclic). 
This motivates their use as parameters in TFP analysis.

We first study the \textit{in-degree} and \textit{out-degree} of $v^*$ in tournament $D$.
These local parameters measure how many players beat $v^*$ (in-degree) or are beaten by $v^*$ (out-degree).
Unlike global measures, they offer a player-centric perspective.
They are easy to compute and intuitive to interpret.
These parameters have been investigated in other tournament-related problems~\cite{yang2017possible}.
In this paper we prove that TFP is $\FPT$ when parameterized by either the in-degree or the out-degree of $v^*$.

It is not difficult to see that when the favorite $v^*$ lies in no cycle, $v^*$ wins if and only if no player beats $v^*$.
Thus in this case, TFP can be solved in polynomial time.
This observation motivates studying \sfasvnum{} and \sfvsvnum{}.
The sfas-v number is at most the fas number, and the sfvs-v number is at most the sfas-v number.
Both subset parameters are bounded above by the in-degree and the out-degree of $v^*$.
This raises the question of whether TFP is $\FPT$ when parameterized by these smaller parameters.
In this paper we answer this question in the negative by proving that TFP is para-$\NP$-hard under either parameterization.

Extending the focus from $v^*$ to its neighborhoods leads to \sfasinnum{} and \sfasoutnum{}.
These parameters measure the minimum number of arc deletions needed to exclude the in-neighborhood or the out-neighborhood of $v^*$ from cycles.
We consider them separately and jointly to examine how restrictions on the two neighborhoods affect tractability.
Those four subset parameters have also been studied in tournament-related problems~\cite{bai2025almost,jana2026subset}.

Finally we consider instances in which both $D[\nin(v^*)]$ and $D[\nout(v^*)]$ are acyclic.
We investigate the complexity of TFP under this restriction and seek sufficient conditions for the existence of a winning seeding for $v^*$ even when the relevant parameters are unbounded.
These questions concern when a winning seeding can be found efficiently under local restrictions on the match outcomes.

\subsection{Our Contributions}

We substantially extend the parameterized complexity landscape of TFP through a systematic study of parameters centered on the favorite player.
Figure~\ref{fig:diagram} summarizes our main contribution.
Together with the known $\FPT$ results for the fas and fvs numbers~\cite{ramanujan2017rigging,zehavi2023tournament}, our results complete the $\FPT$/para-$\NP$-hard classification of the several individual parameters shown in the figure.
Our algorithms place the in-degree and out-degree in the green region.
Our hardness results place all four subset parameters in the red region.
This classification shows how moving from global measures of acyclicity to local restrictions changes the tractability of TFP.
We develop these results in the following order.

\begin{enumerate}
    \item \textbf{Degree parameterizations.}
    We prove that TFP is $\FPT$ when parameterized by either the out-degree or the in-degree of $v^*$ (Theorems~\ref{the:out-degree-fpt} and~\ref{the:in-degree-fpt}).
    For out-degree $\ell$ we give an algorithm running in $2^{2^\ell}n^{\mathcal{O}(1)}$ time.
    Under the Exponential Time Hypothesis, no algorithm runs in $2^{2^{\ell/c}}n^{\mathcal{O}(1)}$ time for any constant $c>1$ (Theorem~\ref{the:out-degree-lowerbound}).
    The in-degree algorithm uses winning-witness forests to characterize yes-instances.
    It combines color coding with a reduction to a restricted version of \textsc{Subgraph Isomorphism}.

    \item \textbf{Subset parameterizations.}
    We prove that TFP remains $\NP$-hard when the sfas-v number or the sfvs-v number equals $1$ (Theorem~\ref{thm:hardness} and Corollary~\ref{coro:hardness-sfvs-v}).
    The same holds separately for the sfas-in and sfas-out numbers (Theorem~\ref{thm:hardness-2} and Corollary~\ref{coro:hardness-sfas-out}).
    These results establish the four para-$\NP$-hard classifications in Figure~\ref{fig:diagram}.
    Bounding both neighborhood parameters yields an $\FPT$ algorithm with running time $2^{\mathcal{O}(s\log s)}n^{\mathcal{O}(1)}$ (Theorem~\ref{thm:joint-neighborhood-fpt}).
    Here $s$ is the sum of the sfas-in and sfas-out numbers.
    Thus the two neighborhood parameters admit an $\FPT$ algorithm jointly despite being para-$\NP$-hard individually.

    \item \textbf{Acyclic neighborhood structures.}
    We prove that TFP remains $\NP$-complete even when both $D[\nin(v^*)]$ and $D[\nout(v^*)]$ are acyclic (Theorem~\ref{thm:acyclic-neighborhood-hardness}).
        Under this restriction, TFP is $\FPT$ when parameterized by the sfas-v number $k$ with running time $2^{\mathcal{O}(k\log k)}n^{\mathcal{O}(1)}$ (Theorem~\ref{thm:sfasfpt}).
    We further identify sufficient degree conditions that guarantee a winning seeding even when the relevant parameters are unbounded (Theorems~\ref{thm:acyclic-neighborhood-degree} and~\ref{thm:acyclic-neighborhood-strict-degree}).
    In examining these conditions, we identify an eight-player corner case for the maximum-matching lemma of Kim and Vassilevska Williams~\shortcite{kim2015fixing}.
    We therefore prove Theorem~\ref{thm:acyclic-neighborhood-strict-degree} without relying on this lemma.
\end{enumerate}

\paragraph{Organization.}
Section~\ref{sec:preliminaries} introduces the notation and structural tools.
Section~\ref{sec:degree-parameters} presents the degree parameterizations.
Section~\ref{sec:subset-parameters} studies the subset parameters and their tractable combinations.
Section~\ref{sec:acyclic-neighborhoods} establishes hardness with acyclic neighborhoods and develops sufficient degree conditions.
Section~\ref{sec:conclusion} summarizes the results and outlines directions for future work.

%% file: Figures/fig_survey.tex
\def\topspread{-25pt} 
\def\bottomshift{0pt} 

\begin{figure}[t]
\FIGURE
{\begin{minipage}{\linewidth}\centering
\begin{tikzpicture}[
    thick,
    scale=0.8, transform shape,
    texts/.style={rectangle, sharp corners, inner sep=6pt, draw=none, text=black},
    box/.style={texts, draw=black, fill=white, line width=0.8pt},
    arcs/.style={-{Latex[length=5pt, width=3pt]}, black!85, shorten <=1pt, shorten >=1pt, thick},
]
    \node[box, xshift=\bottomshift] (sfas-out) {sfas-out {num.}};
    \node[texts, anchor=south east, inner sep=3pt, xshift=20pt] at (sfas-out.north west) (sfas-out-ref) {\textbf{Coro.}~\ref{coro:hardness-sfas-out}};
    \node[box, anchor=west, xshift=15pt] at (sfas-out.east) (sfas-in) {sfas-in {num.}};
    \node[texts, anchor=south west, inner sep=3pt, xshift=-20pt] at (sfas-in.north east) (sfas-in-ref) {\textbf{Thm.}~\ref{thm:hardness-2}};
    \node[box, anchor=east, xshift=-\topspread-\bottomshift, yshift=60pt] at (sfas-out.west) (fvs) {fvs {num.}};
    \node[texts, anchor=east, xshift=-2pt] at (fvs.west) (fvs-ref) {(Zehavi 2023)};
    \node[box, yshift=35pt] at (fvs) (fas) {fas {num.}}; 
    \node[texts, anchor=east, xshift=-2pt, align=center] at (fas.west) (fas-ref) {(Ramanujan and\\Szeider 2017)};
    \node[box, anchor=west, xshift=15pt, yshift=70pt] at (sfas-in.east) (in-deg) {in-degree};
    \node[texts, anchor=south, inner sep=3pt] at (in-deg.north) (in-deg-ref) {\textbf{Thm.}~\ref{the:in-degree-fpt}};
    \node[box, anchor=west, xshift=10pt] at (in-deg.east) (out-deg) {out-degree};
    \node[texts, anchor=south, inner sep=3pt] at (out-deg.north) (out-deg-ref) {\textbf{Thm.}~\ref{the:out-degree-fpt}};
    \node[box, yshift=-35pt] at ($(sfas-out)!0.5!(sfas-in)$) (sfas) {sfas-v {num.}};
    \node[box, yshift=-35pt] at (sfas) (sfvs) {sfvs-v {num.}};
    \node[texts, anchor=east, inner sep=3pt, align=center] at (sfas.west) (sfas-ref) {\textbf{Thm.}~\ref{thm:hardness}}; 
    \node[texts, anchor=west, inner sep=3pt, align=center] at (sfvs.east) (sfvs-ref) {\textbf{Coro.}~\ref{coro:hardness-sfvs-v}}; 
    \path let \p1 = (fas-ref.west), \p2 = (sfas-out-ref.west) in coordinate (real-left) at ({min(\x1,\x2)}, 0);
    \path let \p1 = (out-deg-ref.east), \p2 = (sfas-in-ref.east), \p3 = (out-deg.east) in coordinate (real-right) at ({max(\x1,\x2,\x3)}, 0);
    \coordinate (top-left) at ($({fas.north west} -| {real-left}) + (-5pt, 5pt)$);
    \coordinate (top-right) at ($({top-left.north} -| {real-right}) + (5pt, 0)$);
    \coordinate (bottom-left) at ($({top-left.north} |- {sfvs.south}) + (0, -5pt)$);
    \coordinate (bottom-right) at ($({bottom-left.south} -| {top-right})$);
    \coordinate (left-mid) at ($(top-left |- sfas-out.north) + (0, 20pt)$);
    \coordinate (right-mid) at  ($({left-mid} -| {top-right})$);
    \node[anchor=south west, inner sep=6pt, text=green!50!black, font=\bfseries] at (left-mid) {$\FPT$};
    \node[anchor=north west, inner sep=6pt, text=red!70!black, font=\bfseries] at (left-mid) {Para-$\NP$-hard};
    \draw[arcs] (sfvs) -- (sfas);
    \draw[arcs] (sfas) -- (sfas-out);
    \draw[arcs] (sfas) -- (sfas-in);
    \draw[arcs] (fvs) -- (fas);
    \draw[arcs] (sfvs) to[out=180, in=230] (fvs);
    \draw[arcs] (sfas-out) to[out=90, in=-30] (fas);
    \draw[arcs] (sfas-in) to[out=90, in=0] (fas);
    \draw[arcs] (sfas) to [out=0, in=265] (in-deg);
    \draw[arcs] (sfas) to [out=0, in=275] (out-deg);
    \begin{scope}[on background layer]
        \fill[green!10]
            ($(top-left) + (8pt, 0)$) 
            -- ($(top-right) + (-8pt, 0)$) arc[start angle=90, end angle=0, radius=8pt]
            -- (right-mid)
            -- (left-mid)
            -- ($(top-left) + (0pt, -8pt)$) arc[start angle=180, end angle=90, radius=8pt]
            -- cycle; 
        \fill[red!10]
            (left-mid)
            -- (right-mid)
            -- ($(bottom-right) + (0pt, 8pt)$) arc[start angle=360, end angle=270, radius=8pt]
            -- ($(bottom-left) + (8pt, 0pt)$) arc[start angle=270, end angle=180, radius=8pt]
            -- cycle;
    \end{scope}
\end{tikzpicture}
\end{minipage}}
{
        An illustration of the hierarchy of the eight parameters, where an arrow from parameter $x$ to parameter $y$ denotes $x \leq y$.
        The green and red regions indicate $\FPT$ and para-$\NP$-hardness respectively.
        The fas and fvs results are due to Ramanujan and Szeider~\citeyearpar{ramanujan2017rigging} and Zehavi~\citeyearpar{zehavi2023tournament} respectively.
    \label{fig:diagram}}
{}
\end{figure}

%% file: Figures/tab_combined_structural_results.tex
\begin{table}[t]
\TABLE
{TFP with combined parameters or acyclic neighborhoods.
Here $s$ is the sum of the sfas-in and sfas-out numbers and $k$ is the sfas-v number.
The number of vertices is $n=|V(D)|$.\label{tab:combined-structural-results}}
{\colorlet{scolor}{black}
\centering
\small
\setlength{\tabcolsep}{4pt}
\renewcommand{\arraystretch}{1.15}
\begin{tabular}{@{}>{\raggedright\arraybackslash}p{.36\linewidth}>{\raggedright\arraybackslash}p{.10\linewidth}>{\raggedright\arraybackslash}p{.32\linewidth}>{\raggedright\arraybackslash}p{.14\linewidth}@{}}
\hline
Structural restriction & Parameter & Complexity & Result\\
\hline
General tournaments & $s$ & $\FPT$ \newline $2^{\mathcal{O}(s\log s)}n^{\mathcal{O}(1)}$ & Theorem~\ref{thm:joint-neighborhood-fpt}\\
Both $D[\nin(v^*)]$ and $D[\nout(v^*)]$ are acyclic & - & $\NP$-complete & Theorem~\ref{thm:acyclic-neighborhood-hardness}\\
Both $D[\nin(v^*)]$ and $D[\nout(v^*)]$ are acyclic & $k$ & $\FPT$ \newline $2^{\mathcal{O}(k\log k)}n^{\mathcal{O}(1)}$ & Theorem~\ref{thm:sfasfpt}\\
\hline
\end{tabular}}
{}
\end{table}

%% file: Sources/ch_combined_Preliminaries.tex
\paragraph{Notation and tournaments.}
Let $[n] = \left\{1, 2, \dots, n\right\}$. In this paper we focus on digraphs (including tournaments) denoted by $D = (V, E)$ with $n=|V|$ vertices and $m=|E|$ arcs, where $V(D)$ and $E(D)$ represent its vertex and arc sets respectively.
Whenever we consider a \emph{rooted tree}, we treat it as a digraph where each arc is directed from the parent to the child. 
We use $(u, v) \in E$ to denote an arc from vertex $u$ to vertex $v$, where $u$ is an \emph{in-neighbor} of $v$ and $v$ is an \emph{out-neighbor} of $u$. 
The \emph{out-neighborhood} and \emph{in-neighborhood} of a vertex $v$ are denoted by $\nout(v)$ and $\nin(v)$  respectively.
For a vertex $v$, we further let $out(v) = |\nout(v)|$ and $in(v) = |\nin(v)|$.
For a vertex $v$ and a vertex subset $S\subseteq V$, let $out_S(v) = |\nout(v)\cap S|$ and $in_S(v) = |\nin(v)\cap S|$.
For a vertex subset $X \subseteq V$, its out-neighborhood and in-neighborhood are defined as $\nout(X) = (\bigcup_{v \in X} \nout(v)) \setminus X$ and $\nin(X) = (\bigcup_{v \in X} \nin(v)) \setminus X$ respectively. We call a vertex $s$ a \emph{source} if $in(s) = 0$.

For a vertex subset $X \subseteq V$, the \emph{subgraph induced by $X$} is denoted by $D[X]$.
We write $D - X$ to denote the subgraph $D[V \setminus X]$.
If graph $F$ satisfies $V(F) \subseteq V(D)$ and $E(F) \subseteq E(D)$, we call it a subgraph of $D$ and say $F \subseteq D$.
For a terminal set $T\subseteq V$, a vertex set $X\subseteq V\setminus T$ is a \emph{subset feedback vertex set} of $D$ with terminal set $T$ if there is no cycle containing vertices in $T$ in $D - X$.
Similarly an arc set $A\subseteq E(D)$ is a \emph{subset feedback arc set} of $D$ with terminal set $T$ if there is no cycle containing vertices in $T$ in $D-A$, where $D-A=(V(D),E(D)\setminus A)$.
We will always use $v^*$ to denote the favorite in the tournament.
The terminal set is $\{v^*\}$ for sfas-v. It is $\nin(v^*)$ for sfas-in and $\nout(v^*)$ for sfas-out. All terminal sets are fixed in the input tournament $D$ before any arc deletions.

A knockout tournament with $n = 2^p$ players for some integer $p\geq 0$ is organized into $\log n$ successive rounds of competition. 
In each round $r \in [\log n]$ the $2^{\log n - r + 1}$ remaining players are paired into $2^{\log n - r}$ matches, with only the winners advancing to the next round.
A seeding $\sigma$ is a \emph{winning seeding} for $v^*$ if $v^*$ survives all $\log n$ rounds of the resulting knockout tournament.

Given a tournament $D$ and a seeding $\sigma$, we also use the following notion to describe the actual matches played in each round in the resulting knockout tournament.

\begin{definition}[match set and match set sequence]\label{def:match}
    Let $D$ be a tournament with $n$ players. For a fixed seeding $\sigma$, we define: 
    \begin{itemize}
        \item \textbf{Match Set:} For each round $r \in [\log n]$ the \emph{match set} $M_r \subseteq E(D)$ is the set of arcs representing matches played in round $r$, where $|M_r| = 2^{\log n - r}$. Each pair $(u, v) \in M_r$ indicates that player $u$ beat $v$. 
        Let $V(M_r)$ denote the set of players who participated in round $r$.
        Each player in $V(M_r)$ occurs in exactly one match of $M_r$.
        \item \textbf{Match Set Sequence:} A \emph{match set sequence} is an ordered collection of match sets for all rounds, defined as $\mm = (M_1, \dots, M_{\log n})$.
    \end{itemize}
    
    A match set sequence $(M_1, \dots, M_{\log n})$ is \emph{valid} for $D$ if the following conditions hold:
    \begin{itemize}
        \item For every round $r \in [\log n]$, $M_r \subseteq E(D)$.
        \item $V(M_1) = V(D)$.
        \item For every round $r \in [\log n - 1]$, the set of winners in $M_r$ is exactly $V(M_{r+1})$.
    \end{itemize}
    
    A seeding $\sigma$ induces a unique match set sequence denoted by $\mm(\sigma)$. 
    In contrast if a match set sequence $\mm^*$ is valid, then there exists at least one seeding $\sigma^*$ such that $\mm(\sigma^*) = \mm^*$.
\end{definition}

To formally track the progression of the tournament, we introduce notation for the sets of winners and eliminated players in each round.
Let $D$ be a tournament with $n$ players, and let $\mm(\sigma) = (M_1, \dots, M_{\log n})$ be the match set sequence induced by a seeding $\sigma$. 
For each round $r \in [\log n]$, we use $C_r(\sigma) = \{u \mid (u, v) \in M_r\}$ to denote the set of players remaining after round $r$. 
By convention, we define $C_0(\sigma)=V(D)$. 
Correspondingly for each $r \in [\log n]$, we use $L_r(\sigma)=\{v \mid (u, v) \in M_r\}$ to denote the set of players eliminated in round $r$. 
Where the seeding $\sigma$ is clear from the context, we simply write $C_r$ and $L_r$.
A seeding $\sigma$ is winning for $v^*$ if and only if $C_{\log n}(\sigma)=\{v^*\}$.

\paragraph{Binomial Arborescence.}
Given a rooted forest $T$ and a vertex $v \in V(T)$, let $T_v$ denote the subtree of $T$ rooted at $v$ consisting of $v$ and all its descendants in $T$. An \emph{arborescence} is a rooted directed tree such that all arcs are directed away from the root. In order to reformulate TFP, we introduce the concept of a \emph{binomial arborescence} (see Figure~\ref{fig:obs}).
\begin{definition}[\wang{unlabeled} binomial arborescence]
    \wang{An \emph{unlabeled binomial arborescence (UBA)}} $T$ is defined recursively as follows (see also \citealp{williams2010fixing}):
    \begin{itemize}
        \item A single vertex $v$ is a \wang{UBA} rooted at $v$.
        \item Given two vertex-disjoint \wang{UBAs} of equal size ($T_u$ rooted at $u$ and $T_v$ rooted at $v$), adding a directed arc from $u$ to $v$ yields a new \wang{UBA} rooted at $u$.
    \end{itemize}
    Let $D$ be a digraph. If $T$ is a subgraph of $D$ with $V(T) = V(D)$, then $T$ is called a \wang{\emph{labeled spanning binomial arborescence (LBA)}} of $D$.
\end{definition}

\peng{
For any integer $j \in \mathbb{N}$, there exists a unique {UBA} with $2^j$ vertices. We give some simple properties of {UBAs}.
}
\begin{observation}\label{obs:children}
    Let $T$ be a \wang{UBA} with $2^p$ vertices and let $v \in V(T)$. Then:
    \begin{itemize}
        \item The subtree $T_v$ of $T$ is itself a \wang{UBA}.

        \item Suppose $v$ has $q$ children $r_1, \dots, r_q$. Then these children can be ordered such that the number of children of each $r_i$ is exactly $i - 1$, and their subtrees $T_{r_1}, \dots, T_{r_q}$ contain $2^0, 2^1, \dots, 2^{q-1}$ vertices respectively.

        \item For any $s \in [q]$, take $v$ as the root and attach only the subtrees $T_{r_1}, \dots, T_{r_s}$ as children. The resulting tree $T'$ is a \wang{UBA} of size $2^s$.
    \end{itemize}
\end{observation}

\begin{figure}[t]
\FIGURE
{\begin{minipage}{\linewidth}\centering
    \begin{tikzpicture}[
        thick,
        scale=0.75, transform shape,
        level distance=28pt,
        arcs/.style={-{Stealth[length=6pt, inset=3pt, round, scale width=1.4]}, black, shorten <=1pt, shorten >=1pt, thick},
        every node/.style={circle, draw, minimum size=15pt, inner sep=0pt},
        texts/.style={draw=none, rectangle, inner sep=8pt},
        edge from parent/.style={draw, arcs},
        level 1/.style={sibling distance=65pt},
        level 2/.style={sibling distance=28pt},
        level 3/.style={sibling distance=20pt},
        every fit/.style={rectangle, draw=black, inner sep=5pt, dashed, rounded corners=4pt},
    ]
        \node (root) {$u$}
            child {node[fill=orange!20] {$v$}
                child[sibling distance=38pt] {node {$r_3$}
                    child {node {}
                        child {node {}}
                    }
                    child {node {}}
                }
                child {node[fill=green!20] {$r_2$}
                    child {node {}}
                }
                child {node[fill=green!20] {$r_1$}}
            }
            child {node {}
                child[sibling distance=20pt] {node {}
                    child {node[] {}}
                }
                child[sibling distance=20pt] {node {}}
            }
            child[sibling distance=0pt] {node {}
                child {node {}}
            }
            child[sibling distance=22pt] {node[] {}};
        \node[fit = (root-1) (root-1-2-1) (root-1-3)] (fit) {};
        \node[anchor=south east, texts] at (fit.south east) {$T'$};
    \end{tikzpicture}
\end{minipage}}
{ $T_u$ is a \wang{UBA} with 16 vertices and the subtree $T'$ formed by $T_{r_1}$, $T_{r_2}$ and $v$ is a \wang{UBA} with $2^2 = 4$ vertices.\label{fig:obs}}
{}
\end{figure}

A seeding of $V(D)$ uniquely corresponds to an \wang{LBA} of $D$. The concept of \wang{LBA} plays a central role in our formulation due to the following proposition, which establishes a direct connection between \wang{LBAs} and TFP.
\begin{proposition}[\citealt{williams2010fixing}]
    \label{prop:sba}
    Let $D$ be a tournament with $v^* \in V(D)$. There is a seeding of these players in $D$ such that $v^*$ wins the resulting knockout tournament if and only if $D$ admits an \wang{LBA} rooted at $v^*$.
\end{proposition}

Proposition~\ref{prop:sba} reduces TFP to the problem of finding an \wang{LBA} rooted at the favorite player $v^*$ in the given tournament $D$. \peng{We will adopt this perspective to solve the problem.}

\paragraph{Graph Isomorphism.} 
Two digraphs $F$ and $D$ are \emph{isomorphic} if there exists a bijective function $f: V(F) \to V(D)$ (called an \emph{isomorphism}) such that for every pair of vertices $u, v \in V(F)$, $(u, v) \in E(F)$ if and only if $(f(u), f(v)) \in E(D)$.
\peng{We will use the following parameterized algorithm for subgraph isomorphism.}
\begin{lemma}[\citealt{gupta2018winning}]\label{lem:algsub}
    Let $F$ and $D$ be digraphs on $n_F$ and $n$ vertices respectively and suppose that the treewidth of the underlying undirected graph of $F$ is $\mathrm{tw}$. Let $c : V(D) \to [n_F]$ be a vertex-coloring of $D$ using $n_F$ colors (not necessarily proper). Given two distinguished vertices $f \in V(F)$ and $d \in V(D)$ there exists an algorithm that decides whether $D$ contains a colorful subgraph isomorphic to $F$ (where the isomorphism maps $f$ to $d$) in time $2^{n_F} \cdot n^{\mathrm{tw} + \mathcal{O}(1)}$ using polynomial space. We denote this algorithm by $\algsub(F, D, f, d, c)$. Moreover, \algsub\ can also return a copy of $F$ in $D$ (if one exists) in the same running time bound. This algorithm can immediately solve TFP in $2^{n} \cdot n^{\mo(1)}$ time, where $n$ is the number of players. This algorithm is denoted by $\algexatfp(D, v^*)$.
\end{lemma}

%% file: Sources/ch_Parameterized_by_the_Out-degree_and_In-degree.tex
In this section, we will assume that $\ell=|\nout(v^*)|$ and $k=|\nin(v^*)|$.
We first give an algorithm and an ETH lower bound for the out-degree parameter.
We then turn to the in-degree parameter.
Its algorithm uses a structural characterization of winning seedings by winning-witness forests.

\subsection{Parameterized by the Out-degree}
We begin with the out-degree of the favorite player $v^*$ as the parameter. This case admits a relatively simple analysis. We first show that TFP is $\FPT$ under this parameterization.
\begin{theorem}\label{the:out-degree-fpt}
    TFP can be solved in $2^{2^\ell} \cdot n^{\mo (1)}$ time, where $n$ is the number of vertices of the input tournament $D$ and $\ell$ is the out-degree of the favorite player in $D$.
\end{theorem}
\begin{proof}
In a knockout tournament, the winner must win $\log n$ matches. Therefore if {$\ell < \log n$}, we can safely report that the input instance is a no-instance.
Otherwise if {$\ell \geq \log n$}, we have $n \leq 2^\ell$. Applying the $2^n \cdot n^{\mo(1)}$-time algorithm in \citet{gupta2018winning} gives us the desired running time.
\end{proof}
\peng{
We next present a complementary lower bound result under the Exponential Time Hypothesis (ETH)~\cite{impagliazzo2001problems}, which states that no 3-SAT algorithm runs in $2^{o(N)}$ time on all $N$-variable instances.
\begin{theorem}\label{the:out-degree-lowerbound}
    Under {ETH} no algorithm solves TFP in $2^{2^{{\ell}/{c}}}\cdot n^{\mo(1)}$ time for any constant $c>1$, where $n$ is the number of vertices of the input tournament $D$ and $\ell$ is the out-degree of the favorite player in $D$.
\end{theorem}
\begin{proof}
    \citet{impagliazzo2001problems} showed that under ETH, no 3-SAT algorithm runs in $2^{o(M)}$ time on all instances with $M$ clauses. We consider a variant of 3-SAT (called 3-SAT-2-ltr) where every literal appears at most twice.
    There is a polynomial-time reduction from \citet[Lemma~2.1]{tovey1984simplified} that transforms an instance of 3-SAT with $M$ clauses to an equivalent 3-SAT-2-ltr instance with $N=\mo(M)$ variables. Thus under ETH, no $2^{o(N)}$-time algorithm exists for 3-SAT-2-ltr.
    
    In the NP-completeness proof for TFP in \citet[Theorem~1]{aziz2014fixing}, the authors actually showed that given an instance $F$ of 3-SAT-2-ltr with $N$ variables in polynomial time one can build an instance of TFP with a favorite player who can win the knockout tournament under some seeding if and only if $F$ is satisfiable. Crucially in the resulting instance of TFP the number of vertices of the tournament is $n\leq 64N$, and the out-degree of the favorite player in the tournament is $\ell=\log n\leq \log N+6$ (see also \citealp[Theorem~3]{kim2015fixing}). 

    Suppose for the sake of contradiction that there exists an algorithm $\mathcal{A}$ that solves TFP in time $2^{2^{{\ell}/{c}}}\cdot n^{\mo(1)}$ for some constant $c>1$. 
    Note that $2^{\ell/c}\leq 2^{(\log N+6)/c}=N^{{1}/{c}}\cdot 2^{{6}/{c}}\in o(N)$ since $c>1$.
    Then with the above reduction due to \citet{aziz2014fixing}, we can use $\mathcal{A}$ to solve 3-SAT-2-ltr in $2^{o(N)}$ time (contradicting ETH).
\end{proof}
We remark that Theorem~\ref{the:out-degree-lowerbound} suggests that the simple algorithm in Theorem~\ref{the:out-degree-fpt} is almost optimal under ETH.
}

\subsection{Parameterized by the In-degree}
This section develops the algorithm for the in-degree parameter. We study the complexity of TFP when parameterized by the in-degree of the favorite player $v^*$ in the input tournament $D$. 
The main result is the following theorem:
\begin{theorem}\label{the:in-degree-fpt}
    For an input tournament $D$ and a favorite player $v^* \in V(D)$, TFP is $\FPT$ when parameterized by the in-degree of $v^*$ in $D$.
\end{theorem}

We will devote the rest of this section to proving this result. The proof proceeds in several steps. We begin by analyzing structural properties of yes-instances, and derive useful constraints on the structure of the input tournament. We then reduce the TFP instance to a carefully constructed instance of a restricted version of the \textsc{Subgraph Isomorphism} problem. Finally we design an $\FPT$ algorithm based on the color coding technique \cite{alon1995color}, which has also been applied to solve the general \textsc{Subgraph Isomorphism} problem \cite{amini2012counting}.

In the following analysis we assume $k < \log n$ (or more strictly $k \cdot 2^k < n$), where $n$ is the number of vertices of the input tournament and $k$ is the in-degree of the favorite player in it. This assumption does not affect the fixed-parameter tractability of our algorithm: when it does not hold, $n$ is already bounded by a function of $k$ and we can solve the problem using a $2^n$-time exact algorithm \cite{gupta2018winning} in $\FPT$ time.

\subsubsection{Structural Properties}

We first define nice rounds and nice seedings to characterize winning seedings for $v^*$.
\begin{definition}[nice rounds and nice seedings]
    For a tournament $D$ with $n$ players and a given seeding $\sigma$, we say that round $r \in [\log n]$ is \emph{nice} if it satisfies either $|L_r(\sigma) \cap \nin(v^*)| > 0$ or $|C_{r-1}(\sigma) \cap \nin(v^*)| = 0$. A seeding $\sigma$ is \emph{nice} if all rounds $r$ under $\sigma$ are nice.
\end{definition}
Intuitively a winning seeding is nice if in every round \peng{at least one in-neighbor of $v^*$ is eliminated if it exists}. We then show that for any yes-instance of TFP with few in-neighbors of $v^*$, there always exists a nice winning seeding for $v^*$.
\begin{lemma}[existence of a nice winning seeding]\label{lem:nice}
    Let $I = \left(D, v^*\right)$ be a yes-instance of TFP, where $|V(D)| = n$ and $k = |\nin(v^*)| < \log n$. Then, there exists a nice winning seeding $\sigma$ for $v^*$.
\end{lemma}

\begin{proof}
    Let $\sigma^*$ be an arbitrary winning seeding for $v^*$ in tournament $D$, and let its corresponding match set sequence be $\mm^* = (M_1^*, \cdots, M_{\log n}^*)$. 
    We will prove this lemma by constructing a nice seeding from $\sigma^*$. 
    
    If $\sigma^*$ is already nice, then the claim holds. Otherwise let $p$ be the last round that is not nice. It satisfies $|L_p(\sigma^*) \cap \nin(v^*)| = 0$ and $|C_{p-1}(\sigma^*) \cap \nin(v^*)| > 0$. Furthermore, we have that $|L_{p+1}(\sigma^*)\cap \nin(v^*)|\geq 1$.

    We now modify the match set sequence after round $p-1$ to ``repair'' the non-nice round $p$, without changing the \emph{niceness} of all rounds $r \geq p+1$. Let $X = C_{p-1}(\sigma^*)$ be the set of players participating in round $p$, and let $W = L_{p}(\sigma^*)$ be the players who were eliminated in round $p$ under $\sigma^*$.

    Let $\sigma_{W}$ be an arbitrary seeding of the subtournament $D[W]$, and let player $w \in W \subseteq \nout(v^*)$ be the winner of this subtournament under this seeding. The corresponding match set sequence is $\mm \left(\sigma_{W}\right) = (M_1', \dots, M_{\log n - p}')$.

    Note that the match set sequence
    \begin{equation*}    
        \left(M_{p+1}^* \cup M_1', \dots, M_{\log n}^* \cup M_{\log n - p}', \left\{(v^*, w)\right\}\right)
    \end{equation*}
    (which is obtained by unioning the match set sequences of $W$ and $X \setminus W$ and adding a new round $\left\{(v^*, w)\right\}$) is a valid match set sequence for $X$ (by the structure of knockout tournaments and Definition~\ref{def:match}). Then, match set sequence
    \begin{equation*}
        \begin{aligned}
            \mm = (
                &M_1^*, 
                \cdots, 
                M^*_{p-1}, \\
                &M^*_{p+1} \cup  M'_1,
                \cdots,
                M^*_{\log n} \cup M'_{\log n - p}, \left\{(v^*, w)\right\}
            )
        \end{aligned}
    \end{equation*}
    is also valid for $V(D)$. The corresponding new complete binary tree $T$ of the modified tournament induced by match set sequence $\mm$ is shown in Figure~\ref{fig:nice-lemma}.
    
    \begin{figure}[t]
\FIGURE
{\begin{minipage}{\linewidth}\centering
        \begin{tikzpicture}[
            thick,
            scale=0.75,
            transform shape,
            texts/.style={rectangle, inner sep=2pt, draw=none, font=\large},
            brace/.style={decorate, decoration={calligraphic brace, raise=4pt, amplitude=6pt, mirror}},
            note/.style={texts, anchor=north, yshift=-12pt}
        ]
            \node[texts] at (60:2) (v-star) {$v^*$};
            \node[texts, xshift=-0.5cm, anchor=east] at (0, 0) (round-p) {round $p$};
            \node[texts, anchor=east] at ({$(round-p.east)$}|-{$(v-star)$}) {round $\log n$};
            \node[texts, xshift=2.5cm] at (v-star) (w) {$w$};
            \node[texts, yshift=0.8cm] at ($(v-star)!0.5!(w)$) (root) {$v^*$};
            \node[texts, yshift=-1cm] at (v-star) {$\vdots$};
            \node[texts, yshift=-1cm] at (w) {$\vdots$};
            \draw (0, 0) -- (v-star) -- (2, 0) -- (0, 0);
            \draw (2.5, 0) -- (w) -- (4.5, 0) -- (2.5, 0);
            \draw (v-star) -- (root) -- (w);
            \draw[brace] (0, 0) --node[note] {$X \setminus W$} (2, 0);
            \draw[brace] (2.5, 0) --node[note] {$W \subseteq \nout(v^*)$} (4.5, 0);
        \end{tikzpicture}
\end{minipage}}
{The union of matches among players from $W, X \setminus W$ and a new combined tournament\label{fig:nice-lemma}}
{}
\end{figure}
    
    Let $\sigma$ be a corresponding seeding for $\mm$. Note that $\sigma$ ensures that $v^*$ wins. Moreover the newly modified round $p$ is now nice by construction, since some players from $\nin(v^*) \cap (X \setminus W)$ are eliminated in matches belonging to $M_{p+1}^*$ and these matches occur in round $p$ now.

    \wang{If $\sigma$ is now a nice seeding, the proof is complete. Otherwise, we repeat this process to ``repair'' the new last non-nice round. Since the index of the last non-nice round strictly decreases in each iteration, the overall iterative procedure must terminate.}

    Thus we eventually obtain a nice winning seeding for $v^*$, as desired. This completes the proof.
\end{proof}

For a yes-instance of TFP the previous lemma guarantees the existence of a nice winning seeding, in which every round either eliminates at least one player from $\nin(v^*)$ or there is no player from $\nin(v^*)$. This directly implies that all players in $\nin(v^*)$ must be eliminated within the first $k$ rounds, leading to the following corollary.

\begin{corollary}\label{cor:eliminate}
    Let $I = (D, v^*)$ be a yes-instance of TFP, where $|V(D)| = n$ and $k = |\nin(v^*)| < \log n$. Then any nice winning seeding $\sigma$ for $v^*$ satisfies $|C_k(\sigma) \cap \nin(v^*)| = 0$. Thus no player from $\nin(v^*)$ remains after the first $k$ rounds under the seeding $\sigma$.
\end{corollary}

Building on the previous lemmas and corollary we now analyze the structural properties of the tournament $D$ in a yes-instance of TFP, which would be useful in the design of our algorithm. 
We begin with the following lemma, which shows that each in-neighbor $b \in \nin(v^*)$ is contained in a small structured subtree of $D$.
\begin{lemma}\label{lem:single-b}
    Let $I = (D, v^*)$ be a yes-instance of TFP, where $n=|V(D)|$ and $k = |\nin(v^*)| < \log n$. 
    Let $\sigma^*$ be a nice winning seeding for $v^*$ and $T$ be the \wang{LBA} corresponding to the seeding $\sigma^*$.
    Then for any vertex $b \in \nin(v^*)$, there exists a directed subgraph $F$ of the \wang{LBA} $T$ such that:
    \begin{itemize}
        \item $F$ is an \wang{LBA} rooted at some vertex $u\in \{v^*\}\cup \nout(v^*)$ spanning a subset $X \subseteq V(D)$ with $|X| = 2^k$.
        \item If $v^* \in V(F)$, then $u = v^*$.
        \item $b \in V(F)$.
    \end{itemize}
\end{lemma}
\begin{proof}
    Let $\sigma^*$ be a nice winning seeding for $v^*$ (whose existence is guaranteed by Lemma~\ref{lem:nice}) and let $\mm$ be the corresponding match set sequence. By Proposition~\ref{prop:sba} there exists an \wang{LBA} $T$ of $D$ rooted at $v^*$ such that for every arc $(u, v) \in E(T)$, it is contained in some match set $M \in \mm$.
    
    Fix any $b \in \nin(v^*)$. Starting from $b$ we traverse upward in $T$ toward the root $v^*$, stopping at the first vertex $u$ such that the subtree $T_u$ rooted at $u$ satisfies $|V(T_u)| \ge 2^k$. That is, $u$ is the closest ancestor of $b$ in $T$ such that $|V(T_u)| \ge 2^k$. Let $v$ be the child of $u$ such that $b\in V(T_{v})$.

    Let $r_1, \dots, r_p$ be the children of $u$ in $T$. By Observation~\ref{obs:children}, we can assume that $|V(T_{r_i})| = 2^{i-1}$ for $i\in [p]$. Since $|V(T_u)| \ge 2^k$, we have $p \geq k$. 
    We define $F$ as the subtree rooted at $u$, containing the subtrees $T_{r_1}, \dots, T_{r_k}$ and arcs from $u$ to each $r_i$ ($i \in [k]$). By the definition it holds that $|V(F)| = |\{u\}\cup V(T_{r_1})\cup \dots\cup V(T_{r_k})| = 1+\sum_{i=1}^k 2^{i-1} = 2^k$, and $F$ is an \wang{LBA} by Observation~\ref{obs:children}.

    Since $u$ has at least $k$ children in $T$, $u$ must have survived at least $k$ rounds under the winning seeding $\sigma^*$. 
    By Corollary~\ref{cor:eliminate} all in-neighbors of $v^*$ are eliminated by the end of the first $k$ rounds, which implies $u \notin \nin(v^*)$. 
    Hence, $u\in \{v^*\}\cup\nout(v^*)$.
    Moreover since $v^*$ is the root of $T$, it holds that if $v^* \in V(F)$ then $u = v^*$.
    
    Finally, we prove $b\in V(F)$. 
    Suppose for the sake of contradiction that $b \in V(T_{r_i})$ for some $i > k$.
    \wang{Then the upward iteration would have terminated earlier at the root of $T_{r_i}$, since its size is already at least $2^k$. Hence it holds that $b \in V(T_{r_i})$ for some $i \leq k$, which implies $b\in V(F)$.}
\end{proof}

\begin{definition}[winning-witness forest]\label{def:wwf}
    Let $D$ be a tournament and $v^* \in V(D)$. A subgraph $F \subseteq D$ is called a \emph{winning-witness forest (WWF)} of $(D, v^*)$ if the following conditions hold:
    \begin{itemize}
        \item $F$ consists of $k$ vertex-disjoint \wang{LBAs}, each spanning exactly $2^k$ vertices from $V(D)$ and rooted at some vertex $u_i \in\nout(v^*) \cup \{v^*\}$ for $i \in [k]$.
        \item $\nin(v^*) \subseteq V(F)$.
        \item If $v^* \in V(F)$, then $v^* = u_i$ for exactly one $i\in[k]$.
    \end{itemize}
\end{definition}

The concept of a WWF is central to our algorithmic approach. Intuitively, a WWF is a subgraph of $D$ that contains all in-neighbors of $v^*$.
\begin{lemma}[existence of WWF]\label{lem:forest}
    Let $I = (D, v^*)$ be an instance of TFP with $n = |V(D)|$ and $k = |\nin(v^*)|$. If $k \cdot 2^k < n$, then $I$ is a yes-instance if and only if there exists a $\wwf$ of $(D, v^*)$.
\end{lemma}

\begin{proof}
    We first consider the forward direction. Assume $I$ is a yes-instance. By Lemma~\ref{lem:nice}, choose a nice winning seeding $\sigma^*$ for $v^*$ and let $T$ be its corresponding \wang{LBA} given by Proposition~\ref{prop:sba}. 
    By Lemma~\ref{lem:single-b} for each in-neighbor $b_i \in \nin(v^*)$, there exists an \wang{LBA} $F_i$ of size $2^k$ as a subgraph of $T$ that contains $b_i$ and whose root lies in $\nout(v^*) \cup \{v^*\}$. Let $\mathcal{F} = \{F_i \mid b_i \in \nin(v^*)\}$ denote the collection of these \wang{LBAs} and let $r_i$ be the root of $F_i$.
    
    Recall the construction of these \wang{LBAs} $F_i, F_j \in \mathcal{F}$. We have the following properties: (1) If $r_i = r_j$, then $F_i = F_j$. (2) If $x \in V(F_i)$ and $x \neq r_i$, then $V(T_x) \subseteq V(F_i)$.
    
    We claim that for any two \wang{LBAs} $F_i, F_j \in \mathcal{F}$, they are either identical or vertex-disjoint. We consider the following two cases: (1) If $r_i = r_j$, then $F_i = F_j$. (2) If $r_i \neq r_j$, assume that $r_i \in V(F_j)$. Then $|V(T_{r_i})| < |V(F_j)| = 2^k$. On the other hand, we have $V(F_i) \subseteq V(T_{r_i})$ and in particular $|V(F_i)| = 2^k \leq |V(T_{r_i})|$. This leads to a contradiction. Therefore $r_i \notin V(F_j)$, and similarly $r_j \notin V(F_i)$. Hence, $F_i$ and $F_j$ are vertex-disjoint.
    \wang{Then, the digraph $F$ where $V(F) = \bigcup_{F_i \in \mathcal{F}}V(F_i)$ and $E(F) = \bigcup_{F_i \in \mathcal{F}}E(F_i)$ forms a forest of at most $k$ vertex-disjoint \wang{LBAs}.} Since some $F_i \in \mathcal{F}$ may be identical, $F$ may contain fewer than $k$ trees. We next complete $F$ by adding additional vertex-disjoint \wang{LBAs} until it contains exactly $k$ \wang{LBAs} each of size $2^k$.

    As $k \cdot 2^k < n$, there are enough vertices in $V(D) \setminus V(F)$ to construct the remaining \wang{LBAs}. We repeatedly select disjoint subsets $X \subseteq V(D) \setminus V(F)$ of size $2^k$ and for each $X$ find an \wang{LBA} $F'$ in $D[X]$. Note that the root of $F'$ lies outside $\nin(v^*)$ since $V(F)$ already contains all of $\nin(v^*)$. If $v^*\in X$, then $v^*$ beats every other vertex of $X$ and so $F'$ is rooted at $v^*$. We add these \wang{LBAs} to $F$ until it contains exactly $k$ \wang{LBAs}.
    The resulting forest $F$ is a $\wwf$ of $(D, v^*)$. %
    
    Next, we consider the reverse direction. Assume $F$ is a $\wwf$ of $(D, v^*)$. We prove $I=(D, v^*)$ is a yes-instance by constructing an \wang{LBA} rooted at $v^*$ and spanning $V(D)$.
    Let $X = V(D) \setminus V(F)$ be the remaining vertices. Since $k \cdot 2^k < n$, we can partition $X$ into exactly $(n / 2^k - k)$ disjoint subsets $X_{k+1}, \dots, X_{n / 2^k}$ each of size $2^k$. For each $i > k$, construct an arbitrary \wang{LBA} $F_i$ of size $2^k$ inside the induced subgraph $D[X_i]$. We now obtain exactly $n / 2^k$ \wang{LBAs} in total (including the original ones in $F$).
    
    By construction these \wang{LBAs} partition all of $V(D)$ and satisfy the following two properties:
    
    (1) All roots lie in $\nout(v^*) \cup \{v^*\}$. For the \wang{LBAs} in $F$, this holds by the definition of $\wwf$. For the additional \wang{LBAs} constructed from the disjoint sets $X_i \subseteq V(D) \setminus V(F)$ (since they do not contain all of $\nin(v^*)$), their roots cannot be in $\nin(v^*)$ and are thus in $\nout(v^*) \cup \{v^*\}$.

    (2) Exactly one \wang{LBA} is rooted at $v^*$. If $v^* \in V(F)$, then by the definition of $\wwf$ it is the root of some $F_i$. If $v^* \notin V(F)$ then since $\nin(v^*) \subseteq V(F)$, none of the additional sets $X_i$ contains any in-neighbor of $v^*$ and thus in the \wang{LBA} formed from the set containing $v^*$ it must be the root.

    We now merge these \wang{LBAs} iteratively using the standard \wang{UBA} generating process: at each step, pair \wang{LBAs} arbitrarily and connect their roots to form larger \wang{LBAs} (using the arc between \wang{these} two roots in $D$). Repeat this $\log(n / 2^k)$ times. In each step, the number of \wang{LBAs} halves and their sizes double. Finally since all roots of these \wang{LBAs} lie in $\nout(v^*) \cup \{v^*\}$, we obtain a single \wang{LBA} spanning all of $V(D)$ rooted at $v^*$. Thus, $I$ is a yes-instance.
\end{proof}

\subsubsection{The Algorithm}
Algorithm~\ref{alg:tfp} presents our algorithm for TFP. Note that Algorithm~\ref{alg:tfp} is a Monte Carlo algorithm with one-sided error with a constant probability of a false negative, and it can be derandomized in a standard way.

\begin{algorithm}[tb]
    \caption{$\algtfp(D, v^*)$}
    \label{alg:tfp}
    \textbf{Input:} A tournament $D$ and a favorite player $v^* \in V(D)$.
    
    \textbf{Parameter:} The number of players $n = |V(D)|$ and the in-degree of $v^*$ in $D$, denoted by $k$.
    
    \textbf{Output:} Does there exist a seeding for $V(D)$ such that the favorite player $v^*$ wins?

    \begin{algorithmic}[1]
        \If{$k \cdot 2^k \ge n$}
            \State \Return $\algexatfp(D, v^*)$
        \EndIf

        \For{$iter = 1$ to $\lceil \me^{k \cdot 2^k - k} \rceil$}
            \LComment{\;\;Step~1. coloring step}
            \State Construct a vertex-coloring $c: V(D) \to [k\cdot 2^k]$ of $D$ as follows:\label{line:construct-c-begin}\;
            \State $c(b_i) \gets i$ for each $b_i \in \nin(v^*)$, $i \in [k]$
            \State $c(v) \gets$ sample uniformly at random from $\{k+1,\allowbreak k+2, \dots, k \cdot 2^k\}$ for every $v \in \nout(v^*) \cup \{v^*\}$\label{line:construct-c-end-end}
            
            \LComment{\;\;Step~2. detecting a colorful WWF}
            \State Construct a digraph $F'$ by creating $k$ \wang{UBAs} $T_1, \allowbreak \dots, \allowbreak T_k$ (each of size $2^k$ and rooted at $r_1, \dots, r_k$), adding a root vertex $f$ and adding arcs $(f, r_i)$ for all $i \in [k]$\label{line:construct-F}
            \State Construct a digraph $D'$ by removing all arcs $(b, v^*)$ with $b \in \nin(v^*)$ from $D$, adding a new vertex $d$ and adding arcs $(d, v)$ for all $v \in \nout(v^*) \cup \{v^*\}$\label{line:construct-D}
            \State Construct a vertex-coloring $c': V(D') \to [k\cdot 2^k+1]$ of $D'$ based on $c$ by adding $c'(d) \gets k \cdot 2^k + 1$\label{line:construct-c-end}
            
            \If{$\algsub(F', D', f, d, c')$}
                \State \Return \textbf{true}
            \EndIf
        \EndFor
        \State \Return \textbf{false}
    \end{algorithmic}

\end{algorithm}

The core idea of our algorithm is to employ the color coding technique to determine the existence of a WWF, which is equivalent to checking whether the instance is a yes-instance by Lemma~\ref{lem:forest}. To apply this technique, we now introduce the notion of a \emph{colorful} WWF. A WWF is called colorful under a vertex-coloring if all vertices in it are colored with pairwise distinct colors. The following lemma provides a lower bound on the probability that a WWF (if it exists) becomes colorful under the coloring described in Step 1 of Algorithm~\ref{alg:tfp}.

\begin{lemma}\label{lem:coloring-prob}
    Let $X\subseteq V(D)$ be a vertex subset of $D$ such that $|X|=k \cdot 2^k$ and $\nin(v^*)\subseteq X$, \wang{where $k = |\nin(v^*)|$}. 
    Let $c$ be a coloring of $V(D)$ by the coloring step of Algorithm~\ref{alg:tfp} (Lines~\ref{line:construct-c-begin}--\ref{line:construct-c-end-end}).
    Then the probability that the vertices in $X$ are colored with pairwise distinct colors is at least $\me^{-t}$ where $t=k \cdot 2^k - k$.
\end{lemma}
\begin{proof}
    \peng{
    The fixed coloring of the $k$ vertices in $\nin(v^*)$ leaves $t = k \cdot 2^k - k$ uncolored vertices in $X$ and a total of $t$ available colors. Let $n = |V(D)|$. When the remaining $n-k$ vertices are colored uniformly at random using these $t$ colors, the total number of possible colorings is $t^{n-k}$. The number of favorable outcomes where the $t$ uncolored vertices in $X$ receive pairwise distinct colors is $t! \cdot t^{n-k-t}$. Therefore, the probability that the $t$ uncolored vertices receive distinct colors is $(t! \cdot t^{n - k - t})/(t^{n - k}) = t! \cdot t^{-t} > \me^{-t}$.
    }
\end{proof}

Next in Step 2 we reduce the problem of detecting a colorful WWF to checking the output of the subroutine $\algsub$ on a carefully constructed instance \wang{(see Figure~\ref{fig:construction-add-vertex})}, as shown in the following lemma.
\begin{lemma}\label{lem:wwf-sub}
    For a tournament $D$ and a player $v^* \in V(D)$ with $n = |V(D)|$ and $k = |\nin(v^*)|$ satisfying $k \cdot 2^k < n$, let $c: V(D) \to [k\cdot 2^k]$ be a vertex-coloring of $D$ with $k\cdot 2^k$ colors such that all vertices in $\nin(v^*)$ are assigned pairwise distinct colors that are different from the colors of all other vertices in $V(D)$.
    Let $F'$ and $D'$ be the digraphs constructed in Lines~\ref{line:construct-F}--\ref{line:construct-c-end} of Algorithm~\ref{alg:tfp}. Let $f$ and $d$ be the corresponding vertices and let $c'$ be the vertex-coloring constructed there. 
    Then, there exists a colorful WWF of $(D, v^*)$ under vertex-coloring $c$ if and only if $\algsub(F', D', f, d, c')$ returns true.
\end{lemma}

\begin{proof}
    We first consider the {forward direction.}
    Assume that there exists a colorful WWF $H$ of $(D, v^*)$ under coloring $c$ consisting of $k$ vertex-disjoint \wang{LBAs} each of size $2^k$. Let $F$ be a forest composed of $k$ \wang{UBAs} each with $2^k$ vertices. Then $F$ is isomorphic to $H$ by Definition~\ref{def:wwf}.
    
    We consider the digraph $D'$. \wang{Vertex} $d$ is connected to all roots of the \wang{LBAs} in $H$, as they all lie in $\nout(v^*) \cup \{v^*\}$. Moreover no arc within $H$ is removed, as $v^*$ can only appear as a root in any \wang{LBA} of $H$. \wang{Define $H'\subseteq D'$ by $V(H')=V(H)\cup\{d\}$ and $E(H')=E(H)\cup\{(d,r):r\text{ is a component root of }H\}$.}

    Similarly, construct $F'$ by adding a vertex $f$ to $F$ and arcs from $f$ to the roots of its \wang{UBAs}. Then $F'$ is isomorphic to $H'$ with the isomorphism mapping $f$ to $d$.

    \begin{figure}[t]
\FIGURE
{\begin{minipage}{\linewidth}\centering
        \begin{tikzpicture}[
            thick,
            scale=0.90, transform shape,
            dot/.style={circle, draw, fill, minimum size=3pt, inner sep=0pt},
            every node/.style={circle, draw, inner sep=2pt, minimum size=10pt},
            texts/.style={fill=none, draw=none, rectangle},
            arcs/.style={-{Stealth[length=6pt, inset=3pt, round, scale width=1.4]}, black, shorten <=2pt, shorten >=2pt, thick},
            every fit/.style={rectangle, dashed, rounded corners=5pt, thick, draw=black, inner sep=5pt}
        ]
            \node[dot] at (0, 5pt) (r1) {};
            \draw (r1) -- ++(-10pt, -25pt) -- ++(20pt, 10pt) -- (r1);
            \coordinate (left) at ($(r1)+(-10pt, -25pt)$);
            \node[dot] at (25pt, 5pt) (r2) {};
            \draw (r2) -- ++(-10pt, -25pt) -- ++(20pt, 10pt) -- (r2);
            \node[texts] at (50pt, -12pt) () {$\cdots$};
            \node[dot] at (70pt, 5pt) (r3) {};
            \draw (r3) -- ++(-10pt, -25pt) -- ++(20pt, 10pt) -- (r3);
            \coordinate (right) at ($(r3)+(10pt, -15pt)$);
            \node[dot, color=red] at (35pt, 40pt) (f) {};
            \node[texts, anchor=north west, red] at (f.south east) {$f$};
            \draw[arcs, red] (f) to[in=80, out=190, looseness=1] (r1);
            \draw[arcs, red] (f) to[in=90, out=225, looseness=1] (r2);
            \draw[arcs, red] (f) to[in=100, out=350, looseness=1] (r3);
            \node[dot, red] at (110pt, 0) (d) {};
            \node[dot] at (135pt, 30pt) (b1) {};
            \node[texts] at (135pt, 18pt) {$\vdots$};
            \node[dot] at (135pt, 0pt) (b2) {};
            \node[dot] at (135pt, -20pt) (b3) {};
            \node[dot] at (175pt, 30pt) (b11) {};
            \node[dot] at (175pt, 15pt) (b22) {};
            \node[texts] at (175pt, 0pt) {$\vdots$};
            \node[dot] at (175pt, -20pt) (b33) {};
            \node[texts, anchor=north] at (b3.south) {$v^*$};
            \node[texts, anchor=south, yshift=8pt] (nout) at (b1.north) {$\nout(v^*)$};
            \node[texts, anchor=south, yshift=8pt] (nin) at (b11.north) {$\nin(v^*)$};
            \node[fit=(b1) (b2)] {};
            \node[fit=(b11) (b33)] {};
            \draw[arcs, red] (d) to[in=190, out=80, looseness=1] (b1);
            \draw[arcs, red] (d) to[in=180, out=0, looseness=1] (b2);
            \draw[arcs, red] (d) to[in=170, out=290, looseness=1] (b3);
            \draw[arcs, red, dashed, shorten >=2pt] (b11) to[in=50, out=180, looseness=1] (b3);
            \draw[arcs, red, dashed, shorten >=2pt] (b22) to[in=25, out=180, looseness=1] (b3);
            \draw[arcs, red, dashed, shorten >=2pt] (b33) to[in=0, out=180, looseness=1] (b3);
            \node[texts, anchor=east, red] at (d.west) {$d$};
            \coordinate (d-right) at ($(b33)+(0, -10pt)$);
            \node[texts, anchor=east, xshift=-5pt] at (nout.west) (D) {$D'$};
            \node[texts, xshift=-110pt] at (D) (F) {$F'$};
            \node[fit=(D) (nin) (d-right)] (fit-D) {};
            \coordinate (f-right-bottom) at ($({d-right}-|{right})$);
            \node[fit=(F) (f-right-bottom)] {};
        \end{tikzpicture}
\end{minipage}}
{The construction of $F'$ and $D'$.\label{fig:construction-add-vertex}}
{}
\end{figure}

    Now consider the coloring $c'$ constructed as in the algorithm. Since $H$ is a colorful WWF under $c$, the augmented graph $H'$ is also colorful under $c'$ by construction. Therefore, $\algsub(F', D', f, d, c')$ returns true.

    Next, we consider the {reverse direction.}
    Assume that for some vertex-coloring $c'$ of digraph $D'$ constructed in the algorithm, we know that $\algsub(F', D', f, d, c')$ returns true. Then there exists a subgraph $H'$ of $D'$ such that $F'$ is isomorphic to $H'$ via an isomorphism \wang{which maps $f$ to $d$} and such that all vertices in $H'$ have distinct colors under $c'$.

    Let $H$ denote the subgraph of $H'$ obtained by removing vertex $d$ and all its incident arcs. Similarly, let $F$ be the subgraph of $F'$ obtained by removing $f$ and its incident arcs. By construction $F$ is a forest of $k$ vertex-disjoint \wang{UBAs} and since $F'$ is isomorphic to $H'$ with the isomorphism \wang{mapping} $f$ to $d$, it follows that $F$ is also isomorphic to $H$.

    It remains to verify that $H$ satisfies the following three defining properties of a WWF:
    (1) Each \wang{LBA} has size $2^k$ and the forest $H$ has $k$ such trees, matching the structure of $F$.
    (2) Each $b_i \in \nin(v^*)$ must be contained in $V(H)$. This holds because each $b_i$ was assigned a unique color from $\{1, \dots, k\}$, and $H'$ must include a vertex of each color to be colorful. Hence all $b_i$ are present in $H$.
    (3) Each \wang{UBA} in $F$ is rooted at some vertex $r_i$, which under the isomorphism maps to a vertex in $D'$ connected from $d$. Since $d$ only connects to vertices in $\nout(v^*) \cup \{v^*\}$, the roots of the \wang{LBAs} in $H$ must lie there as well. Moreover, $(d,v^*)$ is the only arc entering $v^*$ in $D'$. Thus if $v^*\in V(H)$ then its parent in $H'$ is $d$, making it the image of a child of $f$ and a root of a component of $H$.

    Thus $H$ is a $\wwf$ of $(D, v^*)$, completing the proof.
\end{proof}

\peng{By combining the structural characterization of WWF with the correctness of our randomized detection algorithm and the efficiency of its implementation, we are now ready to prove Theorem~\ref{the:in-degree-fpt} via the following lemma:}

\begin{lemma}\label{lemma:time-complexity}
    Let $I = (D, v^*)$ be an instance of TFP, where $n = |V(D)|$ and $k = |\nin(v^*)|$.
    Algorithm~\ref{alg:tfp} returns yes with a constant probability if $I$ is a yes-instance and returns no otherwise in time $(2\me)^t \cdot n^{\mo(1)}$, where $t = k \cdot 2^k - k$.
\end{lemma}
\begin{proof}
    If $k \cdot 2^k \geq n$, we invoke the deterministic algorithm $\algexatfp$ in time $2^n \cdot n^{\mo(1)}$ by Lemma~\ref{lem:algsub}. Since $n \leq k \cdot 2^k$, the bound holds.
    Otherwise if $k \cdot 2^k < n$, let $t = k\cdot2^k - k$. For a yes-instance, a WWF exists by Lemma~\ref{lem:forest}. According to Lemma~\ref{lem:coloring-prob} with probability at least $\me^{-t}$, this WWF becomes colorful during the coloring step. Then in Step 2 the algorithm $\algsub$ correctly identifies the existence of such a colorful WWF in time $2^t \cdot n^{\mo(1)}$, by Lemmas~\ref{lem:wwf-sub} and~\ref{lem:algsub}. Repeating this process $\me^t$ times amplifies the success probability to at least $1 - 1/\me$, resulting in a total running time of $(2\me)^t \cdot n^{\mo(1)}$ as stated in the lemma. For a no-instance the algorithm clearly returns no within the same time bound, completing the proof.
\end{proof}

To derandomize the algorithm we only need to derandomize the coloring step, which can be done in a standard way~\cite{naor1995splitters}. 
Roughly speaking, the random coloring step can be replaced with a deterministic enumeration on a family of $\me^{t + o(t)} \cdot \log n$ colorings. Each call to $\algsub$ still takes time $2^t \cdot n^{\mo(1)}$, so the total deterministic time becomes $(2\me)^{t + o(t)} \cdot n^{\mo(1)}$.

%% file: Sources/ch_Parameterized_by_the_Subset_FAS_Number.tex
The degree parameters give fixed-parameter algorithms for TFP.
We next examine the smaller local measures defined by subset feedback sets.
We first prove hardness for each of the four subset parameters even when its value is one.
We then obtain fixed-parameter algorithms by bounding both neighborhood parameters.
When both neighborhoods induce acyclic subtournaments, we obtain a fixed-parameter algorithm with respect to the sfas-v number.
\subsection{Hardness Cases}
In this section, we investigate the complexity of TFP when parameterized by the sfas-v, sfas-in and sfas-out numbers.
We will show that TFP is para-$\NP$-hard with respect to these parameters, thereby completing the parameterized complexity landscape illustrated in Figure \ref{fig:diagram}.

Before presenting our main results, 
we define a class of restricted instances (called \emph{\res{}} instances) and characterize several properties of winning seedings for a \res{} instance. 

\begin{definition}
    An instance $(D, v^*)$ is \res{} with respect to $a^*\in\nout(v^*)$ and $b^*\in\nin(v^*)$ if 
    \begin{enumerate}
        \item $|\nout(v^*)| = 3|\nin(v^*)| - 1$, where {$|\nin(v^*)| = 2^p$}  for some integer $p \geq 0$.
        \item 
        There is an arc $(a^*, b^*)\in E(D)$.
        \item For each vertex $a\in \nout(v^*)$ and each vertex $b\in \nin(v^*)$ except when $a = a^*$ and $b = b^*$, there is an arc $(b, a)\in E(D)$.
    \end{enumerate}
\end{definition}

\input{Figures/res-ins}

See Figure~\ref{fig:res} for an illustration of a \res{} instance $(D, v^*)$.
We introduce the ``special instance'' because it serves as a gadget utilized in the proofs of both Theorem~\ref{thm:hardness} and Theorem~\ref{thm:hardness-2}. 
Intuitively this instance ensures that under a winning seeding (with the exception of the specific match between $a^*$ and $b^*$), no other player from $\nout(v^*)$ competes against a player from $\nin(v^*)$.
The following lemma formally introduces this property.

\begin{lemma}\label{lem:hardness-properties}
    Consider a \res{} instance $(D, v^*)$. Let $A=\nout(v^*)$, $B=\nin(v^*)$ and $n=|B|=2^p$ for some integer $p\ge0$. Then $|A|=3n-1$ and $|V(D)|=4n$.
    Suppose $(D, v^*)$ is a yes-instance.
    Let $\sigma$ be an arbitrary winning seeding for $v^*$ in tournament $D$ and let its corresponding match set sequence be $\mathcal{M}^* =(M_1^*,\dots, M_{(\log n) + 2}^*)$.
    We have the following properties.
    \begin{enumerate}
        \item $V(M_{(\log n) + 1}^*) \cap \nin(v^*) = \{b^*\}$.
        \item $|V(M_{(\log n) + 1}^*) \cap \nout(v^*)| = 2$ and $a^*\in V(M_{(\log n) + 1}^*)$.
        \item 
        {For each round $r\in [\log n]$, every match in $M_r^*$ involves either two players in $A\cup\{v^*\}$ or two players in $B$.}
    \end{enumerate}

\end{lemma}

\begin{proof}
    Let $x_1$ be the player beaten by $v^*$ in the final. In the semifinals, $v^*$ beats $x_2$ and $x_1$ beats $x_3$.
    Hence $M^*_{(\log n) + 1} = \{(v^*, x_2), (x_1, x_3)\}$ and $M^*_{(\log n) + 2} = \{(v^*, x_1)\}$ for some $x_1,x_2,x_3\in V(D)$.
    Since $v^*$ beats $x_1$ and $x_2$, we know that $x_1$ and $x_2$ are in $\nout(v^*)$. 
    Now we show that $x_3\in \nin(v^*)$.

    For $0\le r\le\log n$, let $B_r=C_r(\sigma)\cap B$ denote the players of $B$ remaining after round $r$. Consider any round $r\in[\log n]$.
    By the definition of $D$ we know that for any player $b\in \nin(v^*)\setminus\{b^*\}$, $b$ can only be beaten by players in $\nin(v^*)$. 
    {If $b^*$ remains in this round, we know there are at most $\lceil\frac{1}{2}|B_{r - 1}\setminus\{b^*\}|\rceil$ vertices in $B_{r - 1}$ that are eliminated. 
    If $b^*$ is eliminated in this round,  we know there are at most $\lfloor\frac{1}{2}|B_{r - 1}\setminus\{b^*\}|\rfloor$ vertices in $B_{r - 1} \setminus\{b^*\}$ that are eliminated. This implies that there are at most $\lfloor\frac{1}{2}|B_{r - 1}\setminus\{b^*\}|\rfloor +1 = \lceil\frac{1}{2}|B_{r - 1}|\rceil$ vertices in $B_{r - 1}$ that are eliminated.
    }
    If $b^*$ was eliminated in an earlier round, every remaining player of $B$ can only lose to another player of $B$ and so at most $\lfloor |B_{r-1}|/2\rfloor$ such players are eliminated.
    Thus for each round $r\in[\log n]$,
    \[
        |L_r(\sigma)\cap \nin(v^*)|\leq \lceil \frac{1}{2}|B_{r - 1}|\rceil,
    \] 
    which means that 
    \[
        |B_{r}| \geq |B_{r - 1}| - \lceil \frac{1}{2}|B_{r - 1}|\rceil = \lfloor \frac{1}{2}|B_{r - 1}|\rfloor.
    \]

    Since $|B_0| = |\nin(v^*)| = n$ and $n = 2^p$ for some integer $p\geq 0$, we have that $|B_{\log n}|\geq 1$
    which means that at least one vertex in $\nin(v^*)$ remains after round $\log n$. 
    Thus, $x_3\in \nin(v^*)$.
    Since $x_1$ beats $x_3$ and $x_1\in \nout(v^*)$, we know that $x_1 = a^*$ and $x_3 = b^*$.
    Thus, the first two properties hold.

    Before proving the third property, we refine the bound.
    Recall that there are at most $\lceil\frac{1}{2}|B_{r - 1}\setminus\{b^*\}|\rceil$ vertices in $B_{r - 1} \setminus\{b^*\}$ that are eliminated in round $r\in [\log n]$.
    Since $b^*$ cannot be eliminated before round $(\log n) + 1$, we indeed have that for each round $r\in [\log n]$
    \[
        |L_r(\sigma)\cap \nin(v^*)|\leq \lfloor\frac{1}{2}|B_{r - 1}|\rfloor,
    \]
    which means that 
    \[
        |B_{r}| \geq |B_{r - 1}| - \lfloor \frac{1}{2}|B_{r - 1}|\rfloor = \lceil \frac{1}{2}|B_{r - 1}|\rceil.
    \]

    Now we show that the third property holds.
    By contradiction, assume there exists a match $(b, a)\in M_{r_0}^*$ for some $r_0\in [\log n]$ such that $a\in \nout(v^*)$ and $b\in \nin(v^*)$.
    Let $r$ be the first round that there exists such a match.
    We have that $|B_{r - 1}| = |\nin(v^*)|/ 2^{r - 1}$ is an even number.
    Since each player in $\nin(v^*)\setminus\{b^*\}$ can only be beaten by $\nin(v^*)$ and $b^*$ remains after round $\log n$, we have that 
    \[
        |B_r| > \frac{1}{2}|B_{r - 1}|.
    \]
    Thus we have that 
    \[
        |B_{\log n}| > \frac{1}{2^{\log n}}|\nin(v^*)| = 1,
    \]
    which means that at least two players in $\nin(v^*)$ remain after round $\log n$.
    However $v^*$, $x_1$ and $x_2$ are not in $\nin(v^*)$. This is a contradiction.
    This completes the proof of the third condition.
\end{proof}

Consider a special yes-instance $(D, v^*)$ and an arbitrary winning seeding $\sigma$.
The three properties stated in Lemma~\ref{lem:hardness-properties} hold for $\sigma$.
The complete binary tree $T$ corresponding to $\sigma$ is shown in Figure~\ref{fig:hardness-property}.
One of the main results of this section is the following theorem:

    \begin{figure}[t]
\FIGURE
{\begin{minipage}{\linewidth}\centering
        \input{Figures/example-special-ins}
\end{minipage}}
{An illustration of the complete binary tree $T$ corresponding to a winning seeding $\sigma$ for a \res{} instance $(D, v^*)$, where $a'$ is a player in $\nout(v^*)$. Here $n=|\nin(v^*)|$. Each bottom subtree has $n$ leaves, and its root enters round $\log n+1$.\label{fig:hardness-property}}
{}
\end{figure}

\begin{theorem}\label{thm:hardness}
    For an input tournament $D$ and a favorite player $v^*\in V(D)$, TFP is $\NP$-hard even when \sfasvnum{} equals 1.
\end{theorem}

\begin{proof}
    We will show a polynomial-time reduction from TFP to TFP with the sfas-v number equal to 1.
    Specifically, consider a TFP instance $(D, v^*)$ with $n = |V(D)|$. 
    {We construct an equivalent instance $(D' = (\{v'\}\cup A\cup B, E'), v')$ as follows.
    
    \begin{itemize}
        \item $D'[A]$ is an acyclic tournament with $3n - 1$ vertices where $a^*\in A$ is the source vertex in $D'[A]$.
        \item $D'[B]$ is a copy of $D$ where $b^*\in B$ is the copy of $v^*$. 
        \item For each $a\in A$, there is  an arc $(v', a)$ in $E'$.
        For each $b\in B$, there is an arc $(b, v')$ in $E'$.
        \item $(a^*, b^*) \in E'$ and for  each vertex $a\in A$ and each vertex $b\in B$ except when $a = a^*$ and $b = b^*$, there is an arc $(b, a)\in E'$.
    \end{itemize}
    }

    Clearly $\nout(v') = A$ and $\nin(v') = B$, and $(D', v')$ is a \res{} instance {with respect to $a^*$ and $b^*$}.
    Now we prove both directions of the equivalence of $(D, v^*)$ and $(D', v')$.

    \textbf{Forward direction.} Assume that there exists a winning seeding $\sigma$ for $v^*$ in tournament $D$ (also for $b^*$ in tournament $D'[B]$) and let its corresponding match set sequence be $\mathcal{M}^* =(M_1,\dots, M_{\log n})$.

    Let $(A_1, A_2, A_3)$ be a partition of $A$ such that $|A_1|=|A_2|=n$, $|A_3|=n-1$ and $a^*\in A_1$. So $|A_1| = |A_2| = |A_3\cup\{v'\}| = |B| = n$.
    {Now we construct four seedings for $D'[A_1]$ and $D'[A_2]$ as well as $D'[A_3\cup \{v'\}]$ and $D'[B]$ respectively, and then obtain a winning seeding for $v'$ in $D'$ based on these four seedings.}

    {
    \begin{enumerate}
        \item Let $\sigma_1$ be an arbitrary seeding in tournament $D'[A_1]$. Then $a^*$ must be the winner of $D'[A_1]$ under $\sigma_1$
        since $a^*$ is a source vertex in $D'[A]$.
        \item Let $\sigma_2$ be an arbitrary seeding in tournament $D'[A_2]$ with winner $a'$.
        \item Let $\sigma_3$ be an arbitrary seeding in tournament $D'[A_3\cup \{v'\}]$. Then $v'$ must be the winner of $D'[A_3\cup \{v'\}]$ under $\sigma_3$
        since for each $a\in A$, there is an arc $(v', a)\in E(D')$.
        \item Recall that $\sigma$ is a winning seeding for $b^*$ in tournament $D'[B]$.
    \end{enumerate}
    For $\sigma_i$ $(i = 1, 2, 3)$, the corresponding match set sequence is defined as $\mathcal{M}(\sigma_i) =(M_1^i,\dots, M_{\log n}^i)$.
    Recall that the corresponding match set sequence of $\sigma$ is $\mathcal{M}^* =(M_1,\dots, M_{\log n})$.
    }

    The match set sequence $\mathcal{M}'$ is then defined as:
    \[
        \mathcal{M}' = (M_1'
        ,M_2',\dots, M_{\log n}',
        \{(v', a'), (a^*, b^*)\}, \{(v', a^*)\}),
    \]
    {where $M_r' = M_r\cup \bigcup_{i\in \{1,2,3\}}M_r^i$ for each $r\in [\log n]$.}
     $\mathcal{M}'$ is obtained by unioning the match set sequences of $B$ and $A_1$ as well as $A_2$ and $A_3\cup \{v'\}$ and adding two new rounds $\{(v', a'), (a^*, b^*)\}, \{(v', a^*)\}$, which is a valid match set sequence for $V(D')$.
    
    Let $\sigma'$ be a corresponding seeding for $\mathcal{M}'$. Note that $\sigma'$ ensures that $v'$ wins.
    Thus, there exists a winning seeding for $v'$ in tournament $D'$.
    
    \textbf{Reverse direction.} Assume that there exists a winning seeding $\sigma'$ for $v'$ in tournament $D'$ and let its corresponding match set sequence be 
    \[
        \mathcal{M}' =(M_1',\dots, M_{(\log n) + 2}').
    \]

    By the third property of Lemma~\ref{lem:hardness-properties} there exists a match set sequence $\mathcal{M} = (M_1,\dots, M_{\log n})$ such that for each round $r\in [\log n]$, $M_r\subseteq M_r'$ and $\mathcal{M}$ is valid for $\nin(v')$. 
    Furthermore by the first property of Lemma~\ref{lem:hardness-properties}, the corresponding seeding $\sigma$ for $\mathcal{M}$ ensures that $b^*$ is the winner.
    Since $B$ is the copy of $D$ and $b^*$ is the copy of $v^*$, we have that $\sigma$ ensures that $v^*$ wins.
    Thus, $(D, v^*)$ is a yes-instance.
\end{proof}

Clearly, the construction ensures that \sfvsvnum{} equals 1. 
We have the following corollary.

\begin{corollary}\label{coro:hardness-sfvs-v}
    For an input tournament $D$ and a favorite player $v^*\in V(D)$, TFP is $\NP$-hard even when \sfvsvnum{} equals $1$.
\end{corollary}

We further consider two larger parameters: \sfasinnum{} and \sfasoutnum{}.
Note that both are not smaller than \sfasvnum{} since each cycle traversing $v^*$ must intersect both $\nin(v^*)$ and $\nout(v^*)$.
Nevertheless, we will show that TFP is $\NP$-hard even when \sfasinnum{} equals 1 or \sfasoutnum{} equals 1.
Notably, the reduction established in Theorem~\ref{thm:hardness} already yields the result for \sfasoutnum{}.
Deleting $(a^*,b^*)$ removes the only arc leaving $A$. Since $D'[A]$ is acyclic, no cycle intersects $A=\nout(v')$. Thus, \sfasoutnum{} equals $1$.
We directly obtain the following corollary.

\begin{corollary}\label{coro:hardness-sfas-out}
    For an input tournament $D$ and a favorite player $v^*\in V(D)$, TFP is $\NP$-hard even when \sfasoutnum{} equals $1$.
\end{corollary}

By modifying the construction used in Theorem~\ref{thm:hardness}, we have the following theorem.

\begin{theorem}\label{thm:hardness-2}
    For an input tournament $D$ and a favorite player $v^*\in V(D)$, TFP is $\NP$-hard even when \sfasinnum{} equals 1.
\end{theorem}
\begin{proof}
    We will show a polynomial-time reduction from TFP to TFP with \sfasinnum{} equal to 1.
    Specifically, consider a TFP instance $(D, v^*)$ with $n = |V(D)|$. 
    {We construct an equivalent instance $(D' = (\{v'\}\cup A'\cup A\cup B, E'), v')$ as follows.
    
    \begin{itemize}
        \item $D'[A']$ is an acyclic tournament with $2n - 1$ vertices.
        \item $D'[B]$ is an acyclic tournament with $n$ vertices where $b^*\in B$ is the source vertex in $D'[B]$.
        \item $D'[A]$ is a copy of $D$ where $a^*\in A$ is a copy of $v^*$. 
        \item For each $x\in A'$ and each $a\in A$, there is an arc $(x,a)$ in $E'$.
        \item For each $a\in A'\cup A$, there is an arc $(v', a)$ in $E'$.
        For each $b\in B$,there is an arc $(b, v')$ in $E'$.
        \item $(a^*, b^*) \in E'$ and for each vertex $a\in A'\cup A$ and each vertex $b\in B$ except when $a = a^*$ and $b = b^*$, there is an arc $(b, a)\in E'$.
    \end{itemize}
    }
    See Figure~\ref{fig:thm:hardness-2-app} for an illustration. 
    Clearly $\nout(v') = A\cup A'$ and $\nin(v') = B$, and $(D', v')$ is a \res{} instance {with respect to $a^*$ and $b^*$}.
    Deleting $(a^*,b^*)$ removes the only arc entering $B$. Since $D'[B]$ is acyclic, no cycle intersects $B=\nin(v')$. Thus, \sfasinnum{} equals $1$.
    Now we prove both directions of the equivalence of $(D, v^*)$ and $(D', v')$.

    \begin{figure}[h]
\FIGURE
{\begin{minipage}{\linewidth}\centering
        \begin{tikzpicture}[
            thick,
            scale=0.90, transform shape,
            dot/.style={circle, draw, fill, minimum size=3pt, inner sep=0pt},
            every node/.style={circle, draw, inner sep=2pt, minimum size=10pt},
            texts/.style={fill=none, draw=none, rectangle},
            arcs/.style={-{Stealth[length=6pt, inset=3pt, round, scale width=1.4]}, black, shorten <=2pt, shorten >=2pt, thick},
            every fit/.style={rectangle, dashed, rounded corners=5pt, thick, draw=black, inner sep=5pt}
        ]
            \node[dot] at (50pt, 50pt) (v-star) {};
            \node[texts, text=black, anchor=south west] at (v-star.north east) {$v'$};
            \node[dot] at (10pt, 30pt) (b1) {};                %
            \node[dot] at (10pt, 15pt) (b2) {};               %
            \node[texts] at (10pt, 0pt) {$\vdots$};            %
            \node[dot] at (10pt, -20pt) (b3) {};              %
            \node[dot] at (10pt, -40pt) (a1) {};                %
            \node[texts] at (10pt, -52pt) {$\vdots$};            %
            \node[dot] at (10pt, -70pt) (a2) {};              %
            \node[dot] at (90pt, -40pt) (b11) {};                %
            \node[texts] at (90pt, -52pt) {$\vdots$};           %
            \node[dot] at (90pt, -70pt) (b22) {};                 %
            \node[texts, anchor=east, xshift=-8pt] at (a1.west) {$a^*$};
            \node[texts, anchor=west, xshift=8pt] at (b11.east) {$b^*$};
            \node[fit=(b1) (b3)] (A-1) {};                            %
            \node[texts, anchor=south west, xshift = +5pt] at (a1.north east) {$A=D$}; 
            \node[fit=(b11) (b22)] (B) {};                          %
            \node[texts, anchor=south, yshift=-18pt] (nin) at (B.south) {$\nin(v^*)$}; 
            \node[texts, anchor=south east, xshift = -5pt] at (b11.north west) {$B$}; 
            \node[fit=(a1) (a2)] (A) {};                          %
            \node[texts, anchor=south, yshift=-18pt] (nout) at (A.south) {$\nout(v^*)$}; 
            \node[texts, anchor=south west, xshift = +5pt] at (b1.north east) {$A'$}; 
            \draw[arcs] (v-star.west) to[out=180, in=90, looseness=1] (A-1.north);
            \draw[arcs] (B.north) to[out=90, in=0, looseness=1] (v-star.east);
            \draw[arcs] (v-star.north west) to[out=150, in=180, looseness=1.5] (A.west);
            \draw[arcs] (B.north) to[out=90, in=0, looseness=1] (A-1.east);
            \draw[arcs] (A-1.west) to[out=180, in=180, looseness=1] (A.west); %
            \draw[arcs, red] (a1) to[out=0, in=180, looseness=1] (b11);
            \draw[arcs] (b11) -- (a2);
            \draw[arcs] (b22) -- (a1);
            \draw[arcs] (b22) -- (a2);
            \draw[arcs] (b11) to[out=285, in=75, looseness=1] (b22);
        \end{tikzpicture}
\end{minipage}}
{An illustration of the constructed tournament $D'$.\label{fig:thm:hardness-2-app}}
{}
\end{figure}

    \textbf{Forward direction.} Assume that there exists a winning seeding $\sigma$ for $v^*$ in tournament $D$ (also for $a^*$ in tournament $D'[A]$) and let its corresponding match set sequence be $\mathcal{M}^* =(M_1,\dots, M_{\log n})$.

    Let $(A_1', A_2')$ be a partition of $A'$ such that $|A_1'| = n$ and $|A_2'| = n - 1$. So $|A_1'| = |A_2'\cup \{v'\}| = |A| = |B| = n$.
    Now we construct four seedings for $D'[A_1']$ and $D'[A_2'\cup\{v'\}]$ as well as $D'[A]$ and $D'[B]$ respectively, and then obtain a winning seeding for $v'$ in $D'$ based on these four seedings.

    {
    \begin{enumerate}
        \item Let $\sigma_1$ be an arbitrary seeding in tournament $D'[A_1']$ with winner $a'$.
        \item Let $\sigma_2$ be an arbitrary seeding in tournament $D'[A_2'\cup \{v'\}]$.
        Then $v'$ must be the winner of $D'[A_2'\cup \{v'\}]$ under $\sigma_2$ since for each $a\in A'\cup A$, there is an arc $(v', a)\in E(D')$.
        \item Let $\sigma_3$ be an arbitrary seeding in tournament $D'[B]$.
        Then $b^*$ must be the winner of $D'[B]$ under $\sigma_3$
        since $b^*$ is a source vertex in $D'[B]$.
        
        \item Recall that $\sigma$ is a winning seeding for $a^*$ in tournament $D'[A]$.
    \end{enumerate}
    For $\sigma_i$ $(i = 1, 2, 3)$, the corresponding match set sequence is defined as $\mathcal{M}(\sigma_i) =(M_1^i,\dots, M_{\log n}^i)$.
    Recall that the corresponding match set sequence of $\sigma$ is $\mathcal{M}^* =(M_1,\dots, M_{\log n})$.
    }

    The match set sequence $\mathcal{M}'$ is then defined as: 
    \[
        \mathcal{M}' = (M_1'
        ,M_2',\dots, M_{\log n}',
        \{(v', a'), (a^*, b^*)\}, \{(v', a^*)\}),
    \]
    {where $M_r' = M_r\cup \bigcup_{i\in \{1,2,3\}}M_r^i$ for each $r\in [\log n]$.
    $\mathcal{M}'$ is obtained by unioning the match set sequences of $A_1'$ and $A_2'\cup \{v'\}$ as well as $A$ and $B$ and adding two new rounds $\{(v', a'), (a^*, b^*)\}, \{(v', a^*)\}$, which is a valid match set sequence for $V(D')$.
    }

    Let $\sigma'$ be a corresponding seeding for $\mathcal{M}'$. Note that $\sigma'$ ensures that $v'$ wins.
    Thus, there exists a winning seeding for $v'$ in tournament $D'$.

    \textbf{Reverse direction.} Assume that there exists a winning seeding $\sigma'$ for $v'$ in tournament $D'$ and let its corresponding match set sequence be 
    \[
        \mathcal{M}' =(M_1',\dots, M_{(\log n) + 2}').
    \]

    For $0\le r\le\log n$, let $A_r=C_r(\sigma')\cap A$ be the players of $A$ remaining after round $r$. In particular, $A_0=A$ and $|A_0|=n$.
    By the definition of $D'$ we know that for any player $a\in A$, $a$ can only beat players in $A\cup\{b^*\}$. By the first property of Lemma~\ref{lem:hardness-properties}, $b^*$ survives the first $\log n$ rounds. Every player of $A$ matched against a player outside $A$ during these rounds therefore loses. 
    Thus for each round $r\in[\log n]$, at least $\lceil\frac{1}{2}|A_{r - 1}|\rceil$ vertices in $A_{r - 1}$ are eliminated
    which means that 
    \[
        |A_{r}| \leq |A_{r - 1}| - \lceil \frac{1}{2}|A_{r - 1}|\rceil = \lfloor \frac{1}{2}|A_{r - 1}|\rfloor.
    \]

    Now we assume there exists a match $(x, a)\in M_{r_0}'$ for some $r_0\in [\log n]$ such that $a\in A$ and $x\in V(D')\setminus A$.
    Let $r$ be the first round in which such a match occurs.
    Before round $r$, every match involving a player of $A$ has both endpoints in $A$. We have that $|A_{r - 1}| = n/ 2^{r - 1}$ is an even number.
    Since each player in $A$ can only beat $A\cup\{b^*\}$ and $b^*$ remains after round $\log n$, every player of $A$ paired outside $A$ loses. The remaining players of $A$ are paired internally, so the number paired outside $A$ is a positive even number and hence at least two. Therefore, 
    \[
        |A_r| \leq \frac{|A_{r - 1}|-2}{2} < \frac{1}{2}|A_{r - 1}|.
    \]
    Applying the preceding upper bound for each subsequent round gives 
    \[
        |A_{\log n}| \leq \frac{|A_r|}{2^{\log n-r}} < \frac{|A|}{2^{\log n}} = 1,
    \]
    which means that no player of $A$ remains after round $\log n$.
    However by the second property of Lemma~\ref{lem:hardness-properties}, we know that $a^*\in V(M_{(\log n) + 1}')$. This is a contradiction.
    
    Thus there exists a match set sequence $\mathcal{M} = (M_1,\dots, M_{\log n})$ such that for each round $r\in [\log n]$, $M_r\subseteq M_r'$ and $\mathcal{M}$ is valid for $A$.
    Let $\sigma$ be a corresponding seeding for $\mathcal{M}$.
    Since $a^*$ remains after round $\log n$, we further know that the winner under $\sigma$ is $a^*$.
    Since $A$ is the copy of $D$ and $a^*$ is the copy of $v^*$, we have that $\sigma$ ensures that $v^*$ wins.
\end{proof}

\subsection{Some Tractable Cases}
So far, we have presented some hardness results for TFP. These results can also be understood in the following way. We can partition the tournament apart from vertex \( v^* \) into three parts:
(i) \( D[\nin(v^*)] \),
(ii) \( D[\nout(v^*)] \) and
(iii) the arcs between \( \nin(v^*) \) and \( \nout(v^*) \).

The sfas-in and sfas-out numbers bound the arc deletions needed to make parts (i) and (ii) acyclic respectively. The sfas-v number bounds the arc deletions needed to eliminate cycles through $v^*$. The preceding reductions show that these restrictions alone do not ensure tractability.
The construction in Theorem~\ref{thm:hardness} shows that TFP remains $\NP$-hard when the out-neighborhood induces an acyclic tournament and the sfas-v number equals $1$.
The construction in Theorem~\ref{thm:hardness-2} gives the same hardness when the in-neighborhood induces an acyclic tournament and the sfas-in number equals $1$.

We next examine restrictions on both neighborhoods.
Bounding the sum of the sfas-in and sfas-out numbers yields an $\FPT$ algorithm.
When both neighborhoods are acyclic, a bound on the sfas-v number also suffices.
Theorem~\ref{thm:acyclic-neighborhood-hardness} will show that acyclicity of the two neighborhoods alone does not make TFP polynomial-time solvable unless $\P=\NP$.

First, we show that TFP is $\FPT$ when parameterized by the sum of \sfasinnum{} and \sfasoutnum{}. Our approach uses the following known result about the fvs number.

\begin{lemma}[\citealt{zehavi2023tournament}]\label{lem:fvsfpt}
    TFP is solvable in time $2^{\mathcal{O}(t\log t)}\cdot n^{\mathcal{O}(1)}$ where $n$ and $t$ are the number of vertices and the fvs number, respectively.
\end{lemma}

By applying Lemma~\ref{lem:fvsfpt}, we establish the following result.

\begin{theorem}\label{thm:joint-neighborhood-fpt}
    Given an instance $(D, v^*)$ of TFP the problem can be solved in time $2^{\mathcal{O}(s\log s)}\cdot n^{\mathcal{O}(1)}$, where $n=|V(D)|$ and $s$ is the sum of \sfasinnum{} and \sfasoutnum{}.
\end{theorem}

\begin{proof}
    Let $t$ be the fvs number.
    Let $A_{in}$ be a minimum subset feedback arc set of $D$ with terminal set $\nin(v^*)$.
    Let $A_{out}$ be a minimum subset feedback arc set of $D$ with terminal set $\nout(v^*)$.
    Clearly, $s = |A_{in}|+|A_{out}|$.
    With Lemma~\ref{lem:fvsfpt}, it is sufficient to show that $t \leq s$.
    Since the fas number is not smaller than the fvs number, 
    we will show that $A_{in}\cup A_{out}$ is a feedback arc set of $D$ to complete this proof.

    Consider any cycle $C$ in $D$.
    Clearly $C$ must contain at least one vertex in $\nin(v^*) \cup \nout(v^*)$.
    If $C$ contains one vertex in $\nin(v^*)$, then $C$ contains an arc in $A_{in}$.
    If $C$ contains one vertex in $\nout(v^*)$, then $C$ contains an arc in $A_{out}$.
    Thus, $A_{in}\cup A_{out}$ is a feedback arc set of $D$. The theorem holds.
\end{proof}

Next we show that TFP is $\FPT$ with respect to \sfasvnum{} when both neighborhoods of the favorite induce acyclic subtournaments.

\begin{definition}
An instance $(D, v^*)$ is called a \emph{\DAGins{}} if the induced digraphs $D[\nin(v^*)]$ and $D[\nout(v^*)]$ are both acyclic.
\end{definition}

\begin{theorem}\label{thm:sfasfpt}
    Given a \DAGins{} $(D, v^*)$ of TFP the problem can be solved in time $2^{\mathcal{O}(k\log k)}\cdot n^{\mathcal{O}(1)}$, where $n=|V(D)|$ and $k$ is \sfasvnum{}.
\end{theorem}

\begin{proof}
    Let $t$ be the fvs number.
    Let $A$ be a minimum subset feedback arc set of $D$ with terminal set $\{v^*\}$. Let $V(A)$ denote the set of endpoints of the arcs in $A$.
    With Lemma~\ref{lem:fvsfpt}, it is sufficient to show that $t \leq 2k$.
    Since $|V(A)| \leq 2|A|$,
    we will show that $V(A)$ is a feedback vertex set of $D$ to complete this proof.

    By contradiction, we assume that there exists a cycle $C$ in {$D - V(A)$. Note that $C$ cannot contain $v^*$ since the removal of $A$ from $D$ should leave no cycle containing $v^*$.}
    Since $D[\nin(v^*)]$ is acyclic, $C$ contains at least one vertex in $\nout(v^*)$.
    Since $D[\nout(v^*)]$ is acyclic, $C$ contains at least one vertex in $\nin(v^*)$.
    We further know that $C$ must contain an arc from one vertex $a\in \nout(v^*)$ to one vertex $b\in \nin (v^*)$.
    Since $a,b\notin V(A)$, the cycle $v^*\to a\to b\to v^*$ contains no arc of $A$ and hence remains in $D-A$ {which contradicts that
    $A$ is a subset feedback arc set of $D$ with terminal set $\{v^*\}$.}
\end{proof}

The last result uses both acyclic neighborhoods and a bound on the sfas-v number.
In Section~\ref{sec:acyclic-neighborhoods} we examine what happens when this bound is removed.

%% file: Figures/res-ins.tex
    \begin{figure}[t]
\FIGURE
{\begin{minipage}{\linewidth}\centering
        \begin{tikzpicture}[
            thick,
            scale=0.90, transform shape,
            dot/.style={circle, draw, fill, minimum size=3pt, inner sep=0pt},
            every node/.style={circle, draw, inner sep=2pt, minimum size=10pt},
            texts/.style={fill=none, draw=none, rectangle},
            arcs/.style={-{Stealth[length=6pt, inset=3pt, round, scale width=1.4]}, black, shorten <=2pt, shorten >=2pt, thick},
            every fit/.style={rectangle, dashed, rounded corners=5pt, thick, draw=black, inner sep=5pt}
        ]
            \node[dot] at (50pt, 50pt) (v-star) {};
            \node[texts, anchor=south west] at (v-star.north east) {$v^*$};
            \node[dot] at (20pt, 30pt) (b1) {};                %
            \node[dot] at (20pt, 15pt) (b2) {};               %
            \node[texts] at (20pt, 0pt) {$\vdots$};            %
            \node[dot] at (20pt, -20pt) (b3) {};              %
            \node[dot] at (80pt, 30pt) (b11) {};                %
            \node[texts] at (80pt, 18pt) {$\vdots$};           %
            \node[dot] at (80pt, 0pt) (b22) {};                 %
            \node[texts, anchor=east, xshift=-8pt] at (b1.west) {$a^*$};
            \node[texts, anchor=west, xshift=8pt] at (b11.east) {$b^*$};
            \node[fit=(b1) (b3)] (f1) {};                            %
            \node[texts, anchor=south, yshift=-18pt] (nout) at (f1.south) {$\nout(v^*)$}; 
            \node[fit=(b11) (b22)] (f2) {};                          %
            \node[texts, anchor=south, yshift=-18pt] (nout) at (f2.south) {$\nin(v^*)$}; 
            \draw[arcs] (v-star.west) to[out=180, in=90, looseness=1] (f1.north);
            \draw[arcs] (f2.north) to[out=90, in=0, looseness=1] (v-star.east);
            \draw[arcs, red] (b1) to[out=0, in=180, looseness=1] (b11);
            \draw[arcs] (b11) -- (b2);
            \draw[arcs] (b11) -- (b3);
            \draw[arcs] (b22) -- (b1);
            \draw[arcs] (b22) -- (b2);
            \draw[arcs] (b22) -- (b3);
        \end{tikzpicture}
\end{minipage}}
{An illustration of a \res{} instance $(D, v^*)$.\label{fig:res}}
{}
\end{figure}

%% file: Figures/example-special-ins.tex
  \begin{tikzpicture}[
            thick,
            scale=0.65,
            transform shape,
            texts/.style={rectangle, inner sep=2pt, draw=none, font=\large},
            brace/.style={decorate, decoration={calligraphic brace, raise=4pt, amplitude=6pt, mirror}},
            note/.style={texts, anchor=north, yshift=-12pt}
        ]
            \node[texts] at (60:2) (v-star) {$v^*$};
            \node[texts, xshift=2.5cm] at (v-star) (a1) {$a'$};
            \node[texts, xshift=5cm] at (v-star) (a-star) {$a^*$};
            \node[texts, xshift=7.5cm] at (v-star) (b-star) {$b^*$};
            
            \node[texts, yshift=0.8cm] at ($(v-star)!0.5!(a1)$) (v-star-2) {$v^*$};
            \node[texts, yshift=0.8cm] at ($(a-star)!0.5!(b-star)$) (a-star-2) {$a^*$};
            \node[texts, yshift=0.8cm] at ($(v-star-2)!0.5!(a-star-2)$) (v-star-3) {$v^*$};

            \draw (0, 0) -- (v-star) -- (2, 0) -- (0, 0);
            \draw (2.5, 0) -- (a1) -- (4.5, 0) -- (2.5, 0);
            \draw (5, 0) -- (a-star) -- (7, 0) -- (5, 0);
            \draw (7.5, 0) -- (b-star) -- (9.5, 0) -- (7.5, 0);

            \draw (v-star) -- (v-star-2) -- (a1);
            \draw (a-star) -- (a-star-2) -- (b-star);
            \draw (v-star-2) -- (v-star-3) -- (a-star-2);

            \draw[brace] (0, 0) --node[note] {$\nout(v^*)\cup \{v^*\}$} (7, 0);
            \draw[brace] (7.5, 0) --node[note] {$\nin(v^*)$} (9.5, 0);
        \end{tikzpicture}

%% file: Sources/ch_Acyclic_Neighborhood_Structures.tex
We now consider \DAGins{}s without a bound on the sfas-v number.
We first prove that TFP remains $\NP$-complete on this class.
We then give sufficient degree conditions under which a winning seeding can be found in polynomial time.

\input{Sources/ch_Acyclic_Neighborhood_Hardness.tex}

\subsection{\texorpdfstring{Sufficient Conditions}{Sufficient Conditions}}\label{sec:acyclic-sufficient}

A player $v$ is a \emph{king} if $v$ has distance at most 2 to every other player in the tournament digraph (i.e. for each $u$, $v$ either beats $u$ or
beats another $w$ that beats $u$).
Similarly, a \emph{3-king} is a player that has distance at most 3 to every other player.
A significant research direction involves identifying broad classes of instances that admit efficient solutions.
Prior work has established several sufficient conditions under which a given player can be guaranteed to win \cite{williams2010fixing,stanton2011rigging,kim2015fixing,kim2017can}.
Furthermore the corresponding winning seedings can be constructed in polynomial time, which implies that the winner is susceptible to manipulation in many practical settings.
A substantial portion of these established conditions focuses on players who are kings.

Consider a \DAGins{} $(D, v^*)$ where $D$ contains a source.
If $v^*$ is the unique source, the instance is trivially a yes-instance.
Otherwise it is trivially a no-instance.
Consequently, we focus on the case where there is no source in $D$.
We first establish that in this case, $v^*$ is a 3-king.
\begin{lemma}
    Let $(D, v^*)$ be a \DAGins{} where there is no source in $D$. 
    Then, $v^*$ is a 3-king.
\end{lemma}

\begin{proof}
    For each vertex $a$ in $\nout(v^*)$, we know that the distance from $v^*$ to $a$ is 1.
    Let $b^*$ be the source of $D[\nin(v^*)]$.
    Since $b^*$ is not a source of the whole tournament $D$, there exists a vertex $a'$ in $\nout(v^*)$ such that $(a', b^*)\in E(D)$.
    Thus, we know that the distance from $v^*$ to $b^*$ is 2.
    Since $b^*$ is the source of $D[\nin(v^*)]$, for each vertex $b\in \nin(v^*)$ we know that the distance from $v^*$ to $b$ is at most 3.
    Thus, $v^*$ is a 3-king.
\end{proof}

For comparison, we recall the following statement reported in~\citet[Theorem~5]{kim2015fixing}.

\begin{lemma}[\citealt{kim2015fixing}]\label{lem:3king-existence}
    Let $(D, v^*)$ be an instance where $v^*$ is a 3-king. 
    Let $A = \nout(v^*), B = \nout(A)\cap \nin(v^*)$ and $C = \nin(v^*)\setminus B$.
    Then, $(D, v^*)$ is guaranteed to be a yes-instance if the following conditions hold.
    \begin{enumerate}
        \item $|A|\geq |V(D)|/3$.
        \item $\forall b\in B, out(b)\leq out(v^*)$.
        \item There is a perfect matching from $B$ onto $C$.
    \end{enumerate}
    Furthermore, a winning seeding for $v^*$ can be found in polynomial time.
\end{lemma}

Its published proof invokes Lemma~\ref{lem:king-existence}, which has a corner case discussed below. We use only the hypotheses of Lemma~\ref{lem:3king-existence} for the comparison. Our proofs do not rely on its conclusion.

Consider a \DAGins{} $(D, v^*)$ where $D$ has no source and the second condition of Lemma~\ref{lem:3king-existence} holds.
Let $A = \nout(v^*)$ and $B = \nout(A)\cap \nin(v^*)$.
Let $b^*$ be the source in $D[\nin(v^*)]$.
Since $b^*$ is not a source in $D$, we know that $b^*\in \nout(A)$.
Thus, $b^*\in B$.
By the second condition of Lemma~\ref*{lem:3king-existence}, we know that $out(b^*)\leq out(v^*) = |A|$.
Note that since $b^*$ is a source in $D[\nin(v^*)]$, its out-degree is at least $|\nin(v^*)| - 1$ plus $|\{v^*\}|$ yielding $out(b^*) \geq |\nin(v^*)|$.
Thus we have that $|\nin(v^*)|\leq |\nout(v^*)|$, which implies that $|A|\geq |V(D)|/2$ since $|V(D)|$ is even.
Thus, these hypotheses exclude instances with $|A|<|V(D)|/2$. The following theorem gives a sufficient condition that also applies below this threshold.

\begin{theorem}\label{thm:acyclic-neighborhood-degree}
    Let $(D, v^*)$ be a \DAGins{} such that $D$ has no source. 
    Let $A = \nout(v^*)$ and $B = \nin(v^*)$.
    Then $(D, v^*)$ is guaranteed to be a yes-instance if the following conditions hold.
    \begin{enumerate}
        \item $|A| \geq |V(D)|/3$.
        \item $\forall b\in B, out(b) \leq \frac{|\nin(v^*)|}{|\nout(v^*)|}out(v^*)$.
    \end{enumerate}
\end{theorem}

\begin{proof}
We prove the theorem by induction on the number of vertices $n = |V(D)|$.
To this end for any induced subgraph $D'$ of $D$ containing $v^*$, we define an invariant property $\mathcal{I}(D')$ consisting of the following two conditions:
\begin{enumerate}[label=(\roman*)]
    \item $|A'| \geq |V(D')|/3$, where $A' = \nout(v^*) \cap V(D')$.
    \item For all $b \in B' = \nin(v^*) \cap V(D')$, $out_{V(D')}(b) \leq |B'|$.
\end{enumerate}

Our inductive claim is that any induced subgraph $D'$ containing $v^*$ whose order is a power of two and which satisfies $\mathcal{I}(D')$ gives a yes-instance $(D', v^*)$.

For the base case $n=2$ since $|A'| \geq n/3$, we have $|A'| \geq 1$ meaning $v^*$ beats the only other player.

For the inductive step, assume that the claim holds for all tournaments of size $n/2$. Now consider an arbitrary tournament of size $n$ that satisfies the invariant. We denote this tournament by $D'$. For simplicity of notation we write $in'(\cdot)$ and $out'(\cdot)$ for the in-degree and out-degree within $D'$ (i.e. $in_{V(D')}(\cdot)$ and $out_{V(D')}(\cdot)$), and let $A' = \nout(v^*) \cap V(D')$ and $B' = \nin(v^*) \cap V(D')$.

If $B'=\varnothing$, then $v^*$ beats every other remaining player and wins under any seeding. Henceforth assume $B'\ne\varnothing$.

The main idea is to construct a match set $M$ on $V(D')$ to serve as the first round of the knockout tournament. We will then show that the set of winners of this round (denoted by $V_{\rm win}$) satisfies the invariant $\mathcal{I}(D'[V_{\rm win}])$. By the inductive hypothesis, this implies that $v^*$ wins in the smaller tournament $D'[V_{\rm win}]$ and consequently wins in $D'$.

Before presenting the construction of the match set, we first derive a useful bound for the vertices in $B'$.
Let $A' = \{a_1, \dots, a_{|A'|}\}$ and $B' = \{b_1, \dots, b_{|B'|}\}$ be topologically ordered such that $(a_i, a_j) \in E(D')$ and $(b_i, b_j) \in E(D')$ for all $i < j$.
Consider any vertex $b_i \in B'$. We have
$
    out'(b_i) = out'_{A'}(b_i) + out'_{B'}(b_i) + |\{v^*\}| = out'_{A'}(b_i) + (|B'| - i) + 1
$.
Since $out'(b_i) \leq |B'|$ (Condition (ii) of the invariant), we have $out'_{A'}(b_i) + |B'| - i + 1 \leq |B'|$ which simplifies to
\begin{equation}\label{eq:out_A_bound}
    out'_{A'}(b_i) \leq i - 1.
\end{equation}
We distinguish two cases based on the parity of $|A'|$ and $|B'|$.

\textbf{Case 1.} $|B'|$ is odd and $|A'|$ is even.

Let $k = (|B'| + 1)/2$. From Equation~\eqref{eq:out_A_bound}, we have $out'_{A'}(b_k) \leq k - 1$.
Consequently the in-degree of $b_k$ from $A'$ satisfies:
\begin{align*}
    in'_{A'}(b_k) &= |A'| - out'_{A'}(b_k) \\
              &\geq |A'| - (|B'|-1)/{2} \\
              &= |A'| - {((n - |A'| - 1) - 1)}/{2} \\
              &= {(3|A'| - n)}/{2} + 1.
\end{align*}

Since $|A'| \geq n/3$, we have $3|A'| \ge n$ which implies $in'_{A'}(b_k) \ge 1$.
Thus, there exists a vertex $a' \in A'$ such that $(a', b_k) \in E(D')$. 
Since $|A'|$ is positive and even, choose $a_0\in A'\setminus\{a'\}$.
Let $M_A$ be an arbitrary perfect match set on $A'\setminus\{a',a_0\}$ and let $M_B=\bigcup_{i=1}^{k-1}\{(b_i,b_{i+k})\}$. Construct
\[
M=\{(v^*,a_0),(a',b_k)\}\cup M_A\cup M_B.
\]

Clearly the set of winners is $V_{\rm win} = \{v^*\} \cup A_{\rm win} \cup B_{\rm win}$, where $A_{\rm win}$ consists of $a'$ and the winners from $M_A$ (implying $|A_{\rm win}| = |A'|/2$) and $B_{\rm win} = \{b_1, \dots, b_{k-1}\}$.

We show that the invariant $\mathcal{I}(D'[V_{\rm win}])$ holds. Condition (i) holds since $|A_{\rm win}| = |A'|/2 \geq (n/2)/3$.
Now consider Condition (ii). For any $b_i \in B_{\rm win}$, we have $out'_{V_{\rm win}}(b_i) = out'_{A_{\rm win}}(b_i) + out'_{B_{\rm win}}(b_i) + 1$. Note that $out'_{A_{\rm win}}(b_i) \leq out'_{A'}(b_i) \leq i-1$ (by Equation~\eqref{eq:out_A_bound}) and $out'_{B_{\rm win}}(b_i) = |B_{\rm win}| - i$. Thus it holds that $out'_{V_{\rm win}}(b_i) \leq (i-1) + (|B_{\rm win}| - i) + 1 = |B_{\rm win}|$, which proves Condition (ii).

\textbf{Case 2.} $|B'|$ is even and $|A'|$ is odd.

The proof proceeds analogously to Case 1. The main difference is that since $|B'|$ is even we need to identify two arcs from $A'$ to $B'$ (instead of one).

Let $k = |B'|/2$. By Equation~\eqref{eq:out_A_bound}, $out'_{A'}(b_{k+1}) \leq k$.
Similar to Case 1, we calculate the in-degree of $b_{k+1}$ from $A'$:
\begin{align*}
    in'_{A'}(b_{k+1}) \geq |A'| - {|B'|}/{2} = {(3|A'| - n + 1)}/{2}.
\end{align*}

Since $|A'|$ is odd and $n$ is even, $|A'| \geq n/3$ implies $3|A'| \geq n+1$. Thus $in'_{A'}(b_{k+1}) \geq 1$, which implies there exists a vertex $a' \in A'$ such that $(a', b_{k+1}) \in E(D')$.
Similarly for $b_k$, we have $in'_{A'}(b_k) \geq |A'| - (k-1) = {(3|A'| - n + 3)}/{2}$. Since $3|A'| \geq n+1$, $in'_{A'}(b_k) \geq 2$. This ensures that there exists a vertex $a'' \in A' \setminus \{a'\}$ such that $(a'', b_k) \in E(D')$.
Since $a'$ and $a''$ are distinct and $|A'|$ is odd, $|A'|\ge3$. Choose $a_0\in A'\setminus\{a',a''\}$.
Let $M_A$ be an arbitrary perfect match set on $A'\setminus\{a',a'',a_0\}$ and let $M_B=\bigcup_{i=1}^{k-1}\{(b_i,b_{i+k+1})\}$. Construct
\[
M=\{(v^*,a_0),(a',b_{k+1}),(a'',b_k)\}\cup M_A\cup M_B.
\]
Clearly the set of winners is $V_{\mathrm{win}} = \{v^*\} \cup A_{\mathrm{win}} \cup B_{\mathrm{win}}$, where $A_{\mathrm{win}}$ consists of $a', a''$ and the winners from $M_A$ (implying $|A_{\mathrm{win}}| = (|A'|+1)/2$) and $B_{\mathrm{win}} = \{b_1, \dots, b_{k-1}\}$.

We then check the conditions in $\mathcal{I}(D'[V_{\mathrm{win}}])$. Condition (i) holds since $|A_{\mathrm{win}}| > |A'|/2 \geq (n/2)/3$.
Condition (ii) holds by the same argument as in Case 1 (by applying Equation~\eqref{eq:out_A_bound}).

Since $\mathcal{I}(D'[V_{\mathrm{win}}])$ holds in both cases, by the inductive hypothesis $(D'[V_{\mathrm{win}}], v^*$) is a yes-instance.
Thus, $(D', v^*)$ is a yes-instance.
\end{proof}

Now we consider the case that the favorite player $v^*$ is a king.
We first recall the following statement from~\citet{kim2015fixing}.

\begin{lemma}[\citealt{kim2015fixing}]\label{lem:king-existence}
    Let $(D, v^*)$ be an instance where $v^*$ is a king. 
    Let $k$ be the cardinality of the maximum matching from $\nout(v^*)$ to $\nin(v^*)$.
    Then, $(D, v^*)$ is guaranteed to be a yes-instance if {$|\nout(v^*)|+k > |V(D)|/2$}.
    Furthermore, a winning seeding for $v^*$ can be found in polynomial time.
\end{lemma}

\paragraph{A corner case.}
There is however a corner case in which the hypotheses of Lemma~\ref{lem:king-existence} hold but $v^*$ has no winning seeding. We discuss this case below, where both neighborhoods of $v^*$ induce acyclic tournaments.

Let $V(D)=\{v^*\}\cup A\cup B$, where $A=\{a_1,a_2,a_3\}$ and $B=\{b_1,b_2,b_3,b_4\}$.
Player $v^*$ beats every vertex of $A$ and every vertex of $B$ beats $v^*$. Thus $A=\nout(v^*)$ and $B=\nin(v^*)$.
Within $A$ and $B$, a vertex with a smaller index beats every vertex with a larger index.
The arcs from $A$ to $B$ are exactly
\[
    (a_1,b_3),\quad (a_1,b_4),\quad (a_2,b_2),\quad (a_3,b_1).
\]
All other arcs between $A$ and $B$ are directed from $B$ to $A$.
Each vertex of $B$ is beaten by a vertex of $A$, so $v^*$ is a king.
The matching $\{(a_1,b_3),(a_2,b_2),(a_3,b_1)\}$ has size $3$ and is maximum because $|A|=3$.
Thus with $k=3$, the numerical hypothesis is satisfied: $|A|+k=6>4=|V(D)|/2$.

Suppose that $v^*$ has a winning seeding. It must face all three vertices of $A$ in its three matches.
Its final opponent must therefore win a four-player subtournament containing one vertex of $A$ and three vertices of $B$.
Neither $a_2$ nor $a_3$ can do so because each beats only one vertex of $B$.
The final opponent must then be $a_1$. Any three vertices of $B$ contain $b_1$ or $b_2$.
The one with the smallest index beats $a_1$ and the other two vertices of $B$ and thus wins this subtournament.
This contradiction shows that $v^*$ has no winning seeding.

This corner case also illustrates a limitation of using all maximum-matching arcs in the first round: every maximum matching saturates $A$.
If all three matching arcs are used, the two remaining players are $v^*$ and a vertex of $B$ and so $v^*$ loses immediately.

The following strict-degree statement concerns a \DAGins{} where $v^*$ is a king.

\begin{theorem}\label{thm:acyclic-neighborhood-strict-degree}
    Let $(D, v^*)$ be a \DAGins{} where $v^*$ is a king. 
    Then, $(D, v^*)$ is guaranteed to be a yes-instance if $\forall b\in \nin(v^*), out(b) < 2out(v^*)$.
    Furthermore, a winning seeding for $v^*$ can be found in polynomial time.
\end{theorem}
\begin{proof}
    Let $n=|V(D)|$. Let $A = \nout(v^*)$ and $B = \nin(v^*)$. 
    Williams~\citeyearpar[Theorem~4]{williams2010fixing} shows that a king beating at least half of the players has a winning seeding computable in polynomial time when the number of players is a power of two.
    Since $v^*$ is a king, there should be at least one arc from $A$ to $B$.
    If $|A|\geq n/2$, the cited result gives a winning seeding for $v^*$ in polynomial time.
    Next, we consider the remaining case $|A|<n/2$.

    Let $B = \{b_1, b_2, \dots, b_{|B|}\}$ such that for each $1\leq i < j\leq |B|$, $(b_i, b_j)\in E(D)$. For each $b_i\in B$, we have $in_B(b_i)+out_B(b_i)=|B|-1$ and $in_B(b_i)=i-1$ and thus 
    \begin{align*}
        out(b_i) &= out_A(b_i) + out_B(b_i) + |\{v^*\}|\\
        &= (|A| - in_A(b_i)) + (|B| - 1 - in_B(b_i)) + |\{v^*\}|\\
        &= (|A| - in_A(b_i)) + (|B| - i + 1).
    \end{align*}
    Since $out(b_i) < 2out(v^*) = 2|A|$, we have 
    \[
        (|A| - in_A(b_i)) + (|B| - i + 1) < 2|A|.
    \]
    With $|B| = n - |A| - 1$, we get 
    \begin{align}\label{eq:in_A_bound}
        in_A(b_i) >  n - 2|A| - i.
    \end{align}
    
    Let $p=n-2|A|$.
    Equation~\eqref{eq:in_A_bound} with $i=1$ gives $p\le in_A(b_1)\le |A|$. Since $p$ is positive and even and $p=|A|$ would imply $n=3|A|$ (impossible for a power of two), we have $2\le p\le |A|-1$. Thus $|B|=|A|+p-1<2|A|$ and $p\le |B|$. Also $|A|\ge3$, so $n\ge8$. In particular, $n/4$ is an integer.
    We define the indexing function $\pi: [p]\to [p]$ as
    $\pi(q)=1+n-2|A|-q$. 
    Note that for each $q\in [p]$, it holds that 
    \[in_A(b_{\pi(q)})>n-2|A|-\pi(q)= q-1,\]
    which implies $in_A(b_{\pi(q)})\geq q$.
    
    We construct a matching $M$ of size $k=\min(p, |A|)$ from $A$ to $B$ (where $M = \{(a_q, b_{\pi(q)}) : q \in [k]\}$) by picking an arbitrary $a_q \in \nin(b_{\pi(q)}) \cap A \setminus \{a_1, \dots, a_{q-1}\}$ iteratively. 
    The existence of $a_q$ is guaranteed by a simple counting argument: since $|\nin(b_{\pi(q)})\cap A|=in_{A}(b_{\pi(q)})\geq q$, there always exists an in-neighbor of $b_{\pi(q)}$ in $A$ that has not been matched in the previous $q-1$ steps.

    We then show that $|A|+k>n/2$ holds under $k=\min(p,|A|)=\min(n-2|A|,|A|)$.
    Thus, $k=\min(n-2|A|, |A|)=n-2|A|$. Since $|A| < n/2$, we have $|A| + k = n - |A|>n/2$.
Let $M^*$ be a maximum matching from $A$ to $B$, of size $K$. Then $K\ge k$ and $|A|+K>n/2$. We construct a first round in two cases, writing $A'$ and $B'$ for the surviving vertices of $A$ and $B$.

\textbf{Case 1: $K<|A|$.}
Let $A_U$ and $B_U$ be the vertices of $A$ and $B$ unmatched by $M^*$.
No arc goes from $A_U$ to $B_U$, since it could extend $M^*$. As $v^*$ is a king, every vertex of $B_U$ therefore has an in-neighbor in $A\setminus A_U$.
Use all arcs of $M^*$ and pair $v^*$ with any $a_0\in A_U$.
The remaining sets $A_U\setminus\{a_0\}$ and $B_U$ (of sizes $|A|-K-1$ and $|B|-K$) have the same parity.
Pair their vertices internally, first pairing one vertex from each set if both sizes are odd. The $B_U$ vertex wins this extra match.
All vertices of $A\setminus A_U$ survive, so every vertex of $B'\subseteq B_U$ has an in-neighbor in $A'$ and $v^*$ remains a king. Moreover,
\[
    |A'|=K+\left\lfloor\frac{|A|-K-1}{2}\right\rfloor
    =\left\lfloor\frac{|A|+K-1}{2}\right\rfloor\ge\frac n4,
\]
using $|A|+K\ge n/2+1$.

\textbf{Case 2: $K=|A|$.}
Define $P(a)=\{b\in B:\nin(b)\cap A=\{a\}\}$ for each $a\in A$.
These sets are disjoint and satisfy $\sum_{a\in A}|P(a)|\le |B|<2|A|$.
If some $P(a_0)=\varnothing$, pair $v^*$ with $a_0$ and use all arcs of $M^*$ except the one incident with $a_0$.
Otherwise choose $a_0$ with $P(a_0)=\{\beta\}$. Pair $v^*$ with $a_0$ and pair each $a\in A\setminus\{a_0\}$ with an arbitrary vertex of $P(a)$.
These pairs are disjoint. Since $in_A(b_1)\ge p\ge2$, $b_1$ belongs to no $P(a)$ and so both $b_1$ and $\beta$ remain unpaired. Pair them to eliminate $\beta$, which $b_1$ beats by the ordering of $B$.

In either subcase, all still-unpaired vertices lie in $B$ and number $p$ or $p-2$ respectively. Pair them internally.
All of $A\setminus\{a_0\}$ survives. Only vertices of $P(a_0)$ could lose all their in-neighbors in $A$, and these are absent or eliminated and so $v^*$ remains a king.
Using $p\le |A|-1$ and $|A|\ge3$ gives
\[
    |A'|=|A|-1\ge\frac{2|A|+p}{4}=\frac n4.
\]

Both constructions pair every player exactly once and leave $v^*$ a king in $D'=D[\{v^*\}\cup A'\cup B']$, with $n/2$ vertices and $|A'|\ge n/4$.
Apply Williams's result to $D'$ and expand each leaf of its winning seeding into the corresponding first-round match.
The matching and ordering operations together with the set and pairing operations as well as Williams's construction all take polynomial time.
    This completes the proof.
\end{proof}

\textbf{Remark.}
The bound $\forall b\in \nin(v^*), out(b) < 2out(v^*)$ is tight.
Specifically if we replace this bound with $\forall b\in \nin(v^*), out(b) \leq 2out(v^*)$, then there exists a no-instance.
Consider a \DAGins{} $(D, v^*)$ where $|\nin(v^*)| = |V(D)|/2$ and $|\nout(v^*)| = |V(D)|/2 - 1$.
Take $n=|V(D)|\ge4$ to be a power of two and orient both neighborhoods acyclically.
Let $a'\in \nout (v^*)$ be a specific vertex.
We construct the arc set such that for each vertex $b\in \nin (v^*)$, $(a', b)\in E(D)$ and for each vertex $a\in \nout(v^*)\setminus \{a'\}$ $(b, a)\in E(D)$.
See Figure~\ref{fig:tight} for an illustration.
\input{Figures/example-tight-ins}

Clearly $v^*$ is a king and for each $b\in \nin(v^*)$, $out(b)\leq |\{v^*\}|+(|\nin(v^*)|-1)+(|\nout(v^*)|-1)=2out(v^*)$. 
Nevertheless, $(D, v^*)$ is a no-instance.
Fix any seeding $\sigma$ and let $B_r=C_r(\sigma)\cap\nin(v^*)$ for $0\le r\le\log n$. In each round $a'$ eliminates at most one vertex of $B_{r-1}$, and all other such eliminations occur in matches within $B_{r-1}$. Thus $|B_r|\ge\lfloor|B_{r-1}|/2\rfloor$. Since $|B_0|=n/2$ is a power of two, $|B_{\log n-1}|\ge1$. If $v^*$ reaches the final, it therefore faces a vertex of $\nin(v^*)$ and loses.

%% file: Sources/ch_Acyclic_Neighborhood_Hardness.tex
\begingroup
\subsection{Hardness with Acyclic Neighborhoods}\label{sec:acyclic-hardness}

A TFP instance $(D,v^*)$ is called a \emph{two-chain instance} if it is given together with a partition $V(D)=C_R\cup C_B$ such that $D[C_R]$ and $D[C_B]$ are both acyclic.
Our proof consists of two polynomial-time reductions: from a 3-SAT-2-ltr instance to a two-chain instance and then to a \DAGins{}.
The first reduction uses the construction of Aziz et al.~\citeyearpar[proof of Theorem~1]{aziz2014fixing}.
We additionally prove that its output can be equipped with such a partition in polynomial time.
The second reduction transforms a two-chain instance into a \DAGins{} while preserving the existence of a winning seeding for the favorite player.

\input{Figures/aziz-gadget-drawings}

\paragraph{The construction of Aziz et al.}
The source problem is 3-SAT-2-ltr from Section~\ref{sec:degree-parameters}.
This is an $\NP$-hard variant of 3-SAT in which each literal appears at most twice.
Let $F=(X,\mathcal C)$ be such a formula with $X=\{x_1,\ldots,x_q\}$.
Choose $n$ as the smallest power of two greater than $32q$ and put $k=n/16$.
There are $q$ choice gadgets followed by $q$ clause/garbage gadgets and $k-2q$ filler gadgets.
Their indices are $1,\ldots,q$ followed by $q+1,\ldots,2q$ and $2q+1,\ldots,k$ respectively.
Each gadget $j$ has a player $m_j$ and a set $G_j$ of other players.
The sizes of $G_j$ are $10$ for choice gadgets and $13$ for clause/garbage gadgets.
Each filler gadget has $15$ players in $G_j$.
Write $M=\{m_1,\ldots,m_k\}$ and $H_j=\{m_j\}\cup G_j$.

The remaining players form the set $S=\bigcup_{j=1}^q S_j$.
For each variable $x_j$ we have $S_j=S_{x_j}\cup S_{\bar x_j}\cup\{s_j^*\}$.
Here $S_{x_j}=\{s_{x_j}^1,s_{x_j}^2,s_{x_j}^3\}$ and $S_{\bar x_j}=\{s_{\bar x_j}^1,s_{\bar x_j}^2,s_{\bar x_j}^3\}$.
Following Aziz et al., we call the players in $S_p=\bigcup_{j=1}^q(S_{x_j}\cup S_{\bar x_j})$ \emph{particles}.
The sets $M,S,G_1,\ldots,G_k$ are pairwise disjoint.
The resulting tournament $D_F$ has vertex set $M\cup S\cup\bigcup_{j=1}^kG_j$ and favorite player $m_1$.
Its number of players is $k+7q+10q+13q+15(k-2q)=16k=n$.
In a winning seeding for $m_1$, each $m_j$ wins a 16-player subtournament in the first four rounds.

\begin{table}[H]
\TABLE
{Arc directions within the gadgets of Aziz et al.~\citeyearpar{aziz2014fixing}. For gadget $j$ the subscript $j$ is omitted from all vertex labels, with $m$ denoting $m_j$. The middle column lists the rows from top to bottom and the vertices in each row from left to right. By default each vertex beats every vertex in a lower row and every vertex to its left in the same row.\label{tab:aziz-local-arcs}}
{\colorlet{scolor}{black}
\centering
\small
\setlength{\tabcolsep}{4pt}
\begin{tabular}{@{}p{.15\linewidth}p{.37\linewidth}p{.43\linewidth}@{}}
\hline
Gadget & Rows from left to right & Exceptions to the row rule\\
\hline
Choice &
$(a^1,a^2,b^1,b^2)$\newline
$(x,\bar x)$\newline
$(u)$\newline
$(m)$\newline
$(d,e,c)$ &
$x\to\{a^1,a^2\}$\newline
$\bar x\to\{b^1,b^2\}$\newline
$m\to\{x,\bar x\}$\newline
$e\to m$\\[3pt]
Clause/garbage &
$(m',a',b',c',d',e',g')$\newline
$(m,a,b,c,d,e,g)$ &
$m'\to\{a',b',d'\}$\newline
$b'\to c'$\newline
$d'\to\{e',g'\}$\newline
$m\to\{a,b,d,m'\}$\newline
$b\to c$\newline
$d\to\{e,g\}$\\
\hline
\end{tabular}}
{}
\end{table}

\paragraph{Set $M$.}
We construct $M$ using the \emph{spawning} process of Aziz et al.
Start with a single player.
At each iteration every existing player spawns a new player immediately to its right.
Repeat until there are $k$ players and label them $m_1,\ldots,m_k$ in their final order from left to right.
A player beats the player it spawned.
For every other pair, the player on the right beats the player on the left.
Thus $m_1$ is the initial player.
The parent of a spawned player is the player that spawned it.
Figure~\ref{fig:aziz-construction}(a) illustrates this process for $k=8$.

\paragraph{Global structure.}
Under the \emph{right-left rule} of Aziz et al., every player in $H_j$ beats every player in $H_i$ for $j>i$ unless both players belong to $M$.
The arcs within $M$ are given by the spawning process.
Players in $S$ beat all players outside $S$ unless an exception is specified below.

The particles have a topological order with higher variable indices first.
Within variable $x_j$ we put $S_{\bar x_j}$ before $S_{x_j}$ and order each triple by decreasing superscript.
The players $s_j^*$ also follow decreasing variable index.
Every particle in $S_j$ beats $s_j^*$.
The player $s_j^*$ beats every particle outside $S_j$.
These rules specify all arcs within $S$.

\paragraph{Choice gadgets.}
Gadget $j\le q$ consists of $H_j$ together with $S_j$.
See Figure~\ref{fig:aziz-construction}(b) for an illustration.
Some of its particles can also be used in clause/garbage gadgets.
For the local descriptions below we omit the index $j$ on players.
The symbols $x$ and $\bar x$ denote the two local players corresponding to the literals of variable $x_j$.
Table~\ref{tab:aziz-local-arcs} specifies the arcs within $H_j$ using the row convention of Aziz et al.~\citeyearpar{aziz2014fixing}.
A player in an earlier row beats every player in a later row.
Within a row the player on the right beats the player on the left.
The listed exceptions replace these default directions.
Both $a^1$ and $a^2$ beat every particle in $S_{x_j}$.
Both $b^1$ and $b^2$ beat every particle in $S_{\bar x_j}$.
All other arcs between $S$ and this gadget follow the global rule.

\paragraph{Clause/garbage gadgets.}
For $q<j\le2q$ the gadget consists of $m_j$ and the thirteen players in $G_j$.
See Figure~\ref{fig:aziz-construction}(c) for an illustration.
Its arcs are specified in Table~\ref{tab:aziz-local-arcs}.

The $q$ clause/garbage gadgets contain $2q$ players labeled $g$ or $g'$.
Since $|\mathcal C|\le 4q/3<2q$, we can associate each clause with a distinct one of these players.
The remaining players are garbage players.
A player associated with a clause beats all particles of every literal occurring in that clause.
A garbage player beats all particles.
These are the only exceptions to the rule that $S$ beats this gadget.
In particular, every $s_i^*$ beats both $g$ and $g'$.

\begin{figure}[H]
\FIGURE
{\begin{minipage}{\linewidth}\colorlet{scolor}{black}
\centering
\begin{minipage}[c]{0.45\linewidth}
\centering\color{black}
\AzizSpawning{0}
\endgraf\smallskip (a) Spawning for $k=8$
\endgraf\vspace{12pt}
\AzizClause{0}
\endgraf\smallskip (c) Clause/garbage gadget
\end{minipage}\hfill
\begin{minipage}[c]{0.53\linewidth}
\centering\color{black}
\AzizChoice{0}
\endgraf\smallskip (b) Choice gadget
\end{minipage}
\end{minipage}}
{Construction diagrams adapted from Figures~1--3 of Aziz et al.~\citeyearpar{aziz2014fixing}. In the gadget diagrams unshown arcs run downwards and from right to left within each row.\label{fig:aziz-construction}}
{}
\end{figure}

\paragraph{Filler gadgets.}
For $j>2q$ we start from $m_j$ and apply the same spawning process for four iterations.
The resulting fifteen new players form $G_j$.
This specifies every arc within $H_j$.
Its arcs to other gadgets and to $S$ follow the global rules.
The construction has $n\le64q$ players and can be produced in polynomial time.

The correctness of this construction is the equivalence established by Aziz et al.
We state it as follows.
\begin{lemma}[{\citealp[Lemma~4]{aziz2014fixing}}]\label{lem:aziz-sat-equivalence}
For the formula $F$ and tournament $D_F$ constructed above, $F$ is satisfiable if and only if $m_1$ has a winning seeding in $D_F$.
\end{lemma}

\paragraph{Two-chain instances.}
We now show that the construction yields a two-chain instance.
A topological order lists vertices such that each vertex beats every later vertex.
We write such an order from strongest to weakest using $>$.

\begin{lemma}\label{lem:aziz-two-chains}
A partition $V(D_F)=C_R\cup C_B$ making $(D_F,m_1)$ a two-chain instance with $m_1\in C_R$ can be constructed in polynomial time.
\end{lemma}
\begin{proof}
Recall that $S,H_1,\ldots,H_k$ partition $V(D_F)$, where $H_j=\{m_j\}\cup G_j$.
We construct two color classes $C_R,C_B$ and order each color with its vertices in $S$ first followed by its vertices in $H_k,\ldots,H_1$.
Writing $S_\gamma=S\cap C_\gamma$ for $\gamma\in\{R,B\}$, the proposed order is $S_\gamma>(H_k\cap C_\gamma)>\cdots>(H_1\cap C_\gamma)$.

We color every particle in $S$ red and every player $s_j^*$ blue.
Within $S_\gamma$ use the particle or $s_j^*$ order defined above.
We next specify the coloring and local order within each $H_j$.

We color the players in $M$ by the parity of their depth in the rooted tree defined by the spawning process.
The initial player $m_1$ is red.
Every parent and its child therefore have different colors.
Figure~\ref{fig:aziz-m-coloring} illustrates this coloring and the resulting topological orders on $M$ for $k=8$.

\begin{figure}[H]
\FIGURE
{\begin{minipage}{\linewidth}\colorlet{scolor}{black}
\centering
\input{Figures/aziz-m-coloring}
\end{minipage}}
{Coloring $M$ by depth parity for $k=8$. Panel (a) shows the spawning process with the final labels. Panel (b) shows the same parent-child relations as a rooted tree. Even depths are red and odd depths are blue. Arrows show only parent-child arcs. All other arcs within $M$ run from larger to smaller indices.\label{fig:aziz-m-coloring}}
{}
\end{figure}

We color the remaining players in each choice or clause/garbage gadget using Table~\ref{tab:aziz-two-chains}.
Each row of the table partitions the corresponding $H_j$ and gives the proposed order within each color.
Figure~\ref{fig:aziz-local-colorings} illustrates these colorings using the same gadget layouts.

\begin{table}[htbp]
\TABLE
{The two local topological orders. Local subscripts are omitted and $m$ denotes the gadget player $m_j$. The columns ``Red order'' and ``Blue order'' list the vertices of $H_j\cap C_R$ and $H_j\cap C_B$ respectively in topological order.\label{tab:aziz-two-chains}}
{\colorlet{scolor}{black}
\centering
\small
\setlength{\tabcolsep}{4pt}
\begin{tabular}{@{}lll@{}}
\hline
Gadget and color of $m$ & Red order & Blue order\\
\hline
Choice with red $m$ & $m>\bar x>x>c>d$ & $b^2>b^1>a^2>a^1>u>e$\\
Choice with blue $m$ & $\bar x>x>u>c>e>d$ & $b^2>b^1>a^2>a^1>m$\\
Clause/garbage with red $m$ & $d'>e'>b'>c'>a'>e>c>m$ & $g'>m'>d>g>b>a$\\
Clause/garbage with blue $m$ & $e'>c'>m'>a'>d>e>b>c>a$ & $d'>g'>b'>g>m$\\
\hline
\end{tabular}}
{}
\end{table}

\begin{figure}[htbp]
\FIGURE
{\begin{minipage}{\linewidth}\colorlet{scolor}{black}
\centering
\begin{minipage}[t]{0.48\linewidth}
\centering\color{black}
\AzizChoice{1}
\endgraf\smallskip (a) Choice with red $m$
\endgraf\vspace{12pt}
\AzizClause{1}
\endgraf\smallskip (c) Clause/garbage with red $m$
\end{minipage}\hfill
\begin{minipage}[t]{0.48\linewidth}
\centering\color{black}
\AzizChoice{2}
\endgraf\smallskip (b) Choice with blue $m$
\endgraf\vspace{12pt}
\AzizClause{2}
\endgraf\smallskip (d) Clause/garbage with blue $m$
\end{minipage}
\end{minipage}}
{The local colorings from Table~\ref{tab:aziz-two-chains} on the gadget layouts of Figure~\ref{fig:aziz-construction}. Red vertices belong to $C_R$ and blue vertices belong to $C_B$. All particles are red and $s^*$ is blue in both choice gadgets.\label{fig:aziz-local-colorings}}
{}
\end{figure}

To verify these orders, compare them with the row order in Table~\ref{tab:aziz-local-arcs}.
All arcs between vertices of the same color that follow the default row rule point forwards in the proposed orders, so it remains to check the exceptions.
For a choice gadget with red $m$, the only arcs between vertices of the same color that are exceptions to the row rule are those from $m$ to $\bar x$ and $x$.
There is no such exception for a choice gadget with blue $m$.
For a clause/garbage gadget with red $m$ the exceptions between vertices of the same color are $(d',e'),(b',c'),(d,g)$.
With blue $m$ they are $(m',a'),(d,e),(b,c),(d',g')$.
Every listed direction is an exception in Table~\ref{tab:aziz-local-arcs}.
All remaining pairs of vertices of the same color retain their row direction.
Thus these are topological orders for the two color classes.

We color the vertices in each filler gadget alternately along its rooted tree starting with the prescribed color of $m_j$.
Order each color from right to left in the final spawning order.
No parent-child pair has the same color.
Every arc within a color therefore follows this order.

We now verify that the proposed order is topological for each color.

The two orders inside $S$ are topological.
Between different gadgets all arcs follow decreasing gadget index except possibly those joining two players in $M$.
Two players of $M$ with the same color cannot form a parent-child pair.
Their arc also follows decreasing index by the spawning rule.
Finally, every arc from $H_j$ to $S$ joins a blue player to a red particle.
The relevant blue players are $a^1,a^2,b^1,b^2$ in choice gadgets and $g,g'$ in clause/garbage gadgets.
Hence every arc between $S$ and a gadget within the same color follows the proposed order.
Together with the verification of the orders within each $H_j$, this covers all pairs of vertices of the same color.
The partition and its orders can be constructed in polynomial time.
\end{proof}

Lemmas~\ref{lem:aziz-sat-equivalence} and~\ref{lem:aziz-two-chains} imply that TFP remains $\NP$-hard on two-chain instances.
We use two-chain instances in the following reduction.

\color{black}\paragraph{From two-chain instances to neighbor-acyclic instances.}
\begin{theorem}\label{thm:acyclic-neighborhood-hardness}
TFP is $\NP$-complete on \DAGins{}s.
\end{theorem}
\nopagebreak
\begin{proof}
We show a polynomial-time reduction from TFP on two-chain instances to TFP on \DAGins{}s.
Specifically, consider a two-chain instance $(D,v^*)$ with $n=|V(D)|$ and a given partition $V(D)=C_R\cup C_B$ where $v^*\in C_R$.
We construct an equivalent \DAGins{} $(D',v')$ as follows.

Following the proof of Theorem~\ref{thm:hardness-2}, we construct an auxiliary tournament $\widehat D$ with vertex set $\{v'\}\cup A'\cup A\cup B$.
Let $\widehat D[A']$ be an acyclic tournament with $2n-1$ vertices and let $\widehat D[B]$ be an acyclic tournament with $n$ vertices and source vertex $b^*$.
Let $\widehat D[A]$ be a copy of $D$ where $a^*$ is a copy of $v^*$, and retain the partition $A=C_R\cup C_B$.
All arcs between $A'$ and $A$ are directed from $A'$ to $A$.
There is an arc from $v'$ to every vertex in $A\cup A'$ and an arc from every vertex in $B$ to $v'$.
All arcs between $B$ and $A\cup A'$ are directed from $B$ to $A\cup A'$ except for the arc $(a^*,b^*)$.
Clearly $(\widehat D,v')$ is a \res{} instance with respect to $a^*$ and $b^*$.
We obtain an instance $(D',v')$ with $4n$ vertices by reversing the arcs between $v'$ and $C_B$.
See Figure~\ref{fig:two-chains-construction} for an illustration.

\begin{figure}[htbp]
\FIGURE
{\begin{minipage}{\linewidth}\colorlet{scolor}{black}
\color{black}
\centering
\begin{minipage}[b]{0.48\linewidth}
\centering
\color{black}
\input{Figures/hardness-two-chains-auxiliary}
\endgraf\smallskip
(a) $\widehat D$
\end{minipage}\hfill
\begin{minipage}[b]{0.48\linewidth}
\centering
\color{black}
\input{Figures/hardness-two-chains-final}
\endgraf\smallskip
(b) $D'$
\end{minipage}
\end{minipage}}
{The construction in Theorem~\ref{thm:acyclic-neighborhood-hardness}. (a) The auxiliary tournament $\widehat D$ with $\nout(v')=A'\cup A$ and $\nin(v')=B$. (b) The tournament $D'$ obtained by reversing the arcs between $v'$ and $C_B$. Here $\nout(v')=A'\cup C_R$ and $\nin(v')=B\cup C_B$.\label{fig:two-chains-construction}}
{}
\end{figure}

In $D'$ the out-neighborhood $A'\cup C_R$ of $v'$ has topological order $A'>C_R$.
Its in-neighborhood $B\cup C_B$ has topological order $B>C_B$ because $a^*\in C_R$.
Each set uses its internal topological order, so $(D',v')$ is a \DAGins{}.

\textbf{Forward direction.} Assume that there exists a winning seeding $\sigma$ for $v^*$ in $D$.
Let $(A'_1,A'_2)$ be a partition of $A'$ with $|A'_1|=n$ and $|A'_2|=n-1$.
As in Theorem~\ref{thm:hardness-2}, combine four seedings on $A'_1,A'_2\cup\{v'\},A,B$ with respective winners $a',v',a^*,b^*$.
Use $\sigma$ for the copy on $A$ and arbitrary seedings on the other three sets.
Add the two rounds $\{(v',a'),(a^*,b^*)\}$ and $\{(v',a^*)\}$.
Since $v'$ only meets players in $A'\cup\{a^*\}$ and $a^*\in C_R$, this gives a winning seeding in both $\widehat D$ and $D'$.

\textbf{Reverse direction.} Assume that there exists a winning seeding $\sigma'$ for $v'$ in $D'$.
Since every player in $C_B$ beats $v'$, the seeding $\sigma'$ uses no reversed arc.
Induction over the rounds shows that $\sigma'$ also wins in $\widehat D$ with the same matches and outcomes.
The remaining argument is similar to that in the proof of Theorem~\ref{thm:hardness-2}.

For $0\le r\le\log n$, let $B_r=C_r(\sigma')\cap B$ be the players of $B$ remaining after round $r$ in $\widehat D$.
For $1\le r\le\log n$, at most $\lceil |B_{r-1}|/2\rceil$ players of $B$ are eliminated since only $b^*$ can lose outside $B$.
Hence $|B_r|\ge\lfloor |B_{r-1}|/2\rfloor$ and $|B_{\log n}|\ge1$ since $|B_0|=n$.
Let $(v',x_2)$ and $(x_1,x_3)$ be the semifinals with $x_1$ reaching the final.
Since $v'$ loses to every player in $B$, we have $x_1,x_2\notin B$ and $x_3\in B$.
The only arc entering $B$ from outside is $(a^*,b^*)$, so $x_1=a^*$ and $x_3=b^*$.

For $0\le r\le\log n$, let $A_r=C_r(\sigma')\cap A$ be the players of $A$ remaining after round $r$.
As in Theorem~\ref{thm:hardness-2}, every player of $A$ matched outside $A$ during these rounds loses because $b^*$ reaches the semifinals.
Thus $|A_r|\le\lfloor |A_{r-1}|/2\rfloor$ for $r\in[\log n]$.
Suppose round $r\le\log n$ is the first to contain a match between $A$ and its complement.
The number of players matched outside $A$ is positive and even because $|A_{r-1}|=n/2^{r-1}$ is even.
Consequently $|A_r|\le(|A_{r-1}|-2)/2$ and $|A_{\log n}|\le |A_r|/2^{\log n-r}<1$, contradicting the survival of $a^*$.
Thus the first $\log n$ rounds restrict to a valid match set sequence on $A$ with winner $a^*$.
Since $A$ is a copy of $D$ and $a^*$ is a copy of $v^*$, this yields a winning seeding for $v^*$ in $D$.
\end{proof}
\endgroup

%% file: Figures/aziz-gadget-drawings.tex
\newcommand{\AzizPalette}[1]{%
  \color{black}%
  \ifnum#1=0
    \colorlet{azizR}{black}\colorlet{azizB}{black}%
  \else
    \colorlet{azizR}{red}\colorlet{azizB}{blue}%
  \fi
}
\newcommand{\AzizChoice}[1]{%
\begingroup
\AzizPalette{#1}
\begin{tikzpicture}[
  font=\fontsize{9}{10}\selectfont,
  vertex/.style={circle,draw,fill=white,minimum size=15pt,inner sep=1.4pt},
  R/.style={draw=azizR,text=azizR},B/.style={draw=azizB,text=azizB},
  arc/.style={-{Stealth[length=4pt,width=4pt]},thick,draw=black,shorten <=1pt,shorten >=1pt},
  note/.style={draw=none,fill=none,inner sep=1pt}
]
\node[vertex,B] (star) at (0,2.75) {$s^*$};
\node[vertex,R] (sx1) at (-2.58,1.80) {$s_x^1$};
\node[vertex,R] (sx2) at (-1.80,1.80) {$s_x^2$};
\node[vertex,R] (sx3) at (-1.02,1.80) {$s_x^3$};
\node[vertex,R] (sn1) at (1.02,1.80) {$s_{\bar x}^1$};
\node[vertex,R] (sn2) at (1.80,1.80) {$s_{\bar x}^2$};
\node[vertex,R] (sn3) at (2.58,1.80) {$s_{\bar x}^3$};
\draw[black,thick] (-1.80,1.80) ellipse (1.30cm and 0.50cm);
\draw[black,thick] (1.80,1.80) ellipse (1.30cm and 0.50cm);
\node[note,anchor=south] at (-2.25,2.35) {$S_x$};
\node[note,anchor=south] at (2.25,2.35) {$S_{\bar x}$};
\node[vertex,B] (a1) at (-2.14,0.65) {$a^1$};
\node[vertex,B] (a2) at (-1.46,0.65) {$a^2$};
\node[vertex,B] (b1) at (1.46,0.65) {$b^1$};
\node[vertex,B] (b2) at (2.14,0.65) {$b^2$};
\draw[black,thick] (-1.80,0.65) ellipse (0.75cm and 0.31cm);
\draw[black,thick] (1.80,0.65) ellipse (0.75cm and 0.31cm);
\node[vertex,R] (x) at (-1.80,-0.40) {$x$};
\node[vertex,R] (xb) at (1.80,-0.40) {$\bar x$};
\ifnum#1=2
  \node[vertex,R] (u) at (0,-1.00) {$u$};
  \node[vertex,B] (m) at (0,-1.70) {$m$};
  \node[vertex,R] (e) at (0,-2.65) {$e$};
\else
  \node[vertex,B] (u) at (0,-1.00) {$u$};
  \node[vertex,R] (m) at (0,-1.70) {$m$};
  \node[vertex,B] (e) at (0,-2.65) {$e$};
\fi
\node[vertex,R] (d) at (-0.78,-2.65) {$d$};
\node[vertex,R] (c) at (0.78,-2.65) {$c$};
\draw[arc] (-1.80,2.30) to[out=90,in=180] (star.west);
\draw[arc] (1.80,2.30) to[out=90,in=0] (star.east);
\draw[arc] (-1.80,0.96) -- (-1.80,1.30);
\draw[arc] (1.80,0.96) -- (1.80,1.30);
\draw[arc] (x.north) -- (-1.80,0.34);
\draw[arc] (xb.north) -- (1.80,0.34);
\draw[arc] (m.west) to[out=180,in=270] (x.south);
\draw[arc] (m.east) to[out=0,in=270] (xb.south);
\draw[arc] (e.north) -- (m.south);
\end{tikzpicture}
\endgroup
}
\newcommand{\AzizClause}[1]{%
\begingroup
\AzizPalette{#1}
\begin{tikzpicture}[
  font=\fontsize{9}{10}\selectfont,
  dot/.style={circle,minimum size=3.6pt,inner sep=0pt},
  R/.style={draw=azizR,fill=azizR,text=azizR},
  B/.style={draw=azizB,fill=azizB,text=azizB},
  arc/.style={-{Stealth[length=4pt,width=4pt]},thick,black,shorten <=2pt,shorten >=2pt}
]
\ifnum#1=2
\node[dot,R,label={[text=azizR]left:$m'$}] (mp) at (0.00,1.20) {};
\node[dot,R,label={[text=azizR]above:$a'$}] (ap) at (0.80,1.20) {};
\node[dot,B,label={[text=azizB]above:$b'$}] (bp) at (1.60,1.20) {};
\node[dot,R,label={[text=azizR]above:$c'$}] (cp) at (2.40,1.20) {};
\node[dot,B,label={[text=azizB]above:$d'$}] (dp) at (3.20,1.20) {};
\node[dot,R,label={[text=azizR]above:$e'$}] (ep) at (4.00,1.20) {};
\node[dot,B,label={[text=azizB]above:$g'$}] (gp) at (4.80,1.20) {};
\node[dot,B,label={[text=azizB]left:$m$}] (m) at (0.00,0.00) {};
\node[dot,R,label={[text=azizR]above:$a$}] (a) at (0.80,0.00) {};
\node[dot,R,label={[text=azizR]above:$b$}] (b) at (1.60,0.00) {};
\node[dot,R,label={[text=azizR]above:$c$}] (c) at (2.40,0.00) {};
\node[dot,R,label={[text=azizR]above:$d$}] (d) at (3.20,0.00) {};
\node[dot,R,label={[text=azizR]above:$e$}] (e) at (4.00,0.00) {};
\node[dot,B,label={[text=azizB]above:$g$}] (g) at (4.80,0.00) {};
\else
\node[dot,B,label={[text=azizB]left:$m'$}] (mp) at (0.00,1.20) {};
\node[dot,R,label={[text=azizR]above:$a'$}] (ap) at (0.80,1.20) {};
\node[dot,R,label={[text=azizR]above:$b'$}] (bp) at (1.60,1.20) {};
\node[dot,R,label={[text=azizR]above:$c'$}] (cp) at (2.40,1.20) {};
\node[dot,R,label={[text=azizR]above:$d'$}] (dp) at (3.20,1.20) {};
\node[dot,R,label={[text=azizR]above:$e'$}] (ep) at (4.00,1.20) {};
\node[dot,B,label={[text=azizB]above:$g'$}] (gp) at (4.80,1.20) {};
\node[dot,R,label={[text=azizR]left:$m$}] (m) at (0.00,0.00) {};
\node[dot,B,label={[text=azizB]above:$a$}] (a) at (0.80,0.00) {};
\node[dot,B,label={[text=azizB]above:$b$}] (b) at (1.60,0.00) {};
\node[dot,R,label={[text=azizR]above:$c$}] (c) at (2.40,0.00) {};
\node[dot,B,label={[text=azizB]above:$d$}] (d) at (3.20,0.00) {};
\node[dot,R,label={[text=azizR]above:$e$}] (e) at (4.00,0.00) {};
\node[dot,B,label={[text=azizB]above:$g$}] (g) at (4.80,0.00) {};
\fi
\foreach \from/\to in {mp/ap,bp/cp,dp/ep,m/a,b/c,d/e,m/mp}
  \draw[arc] (\from) -- (\to);
\foreach \from/\to in {mp/bp,mp/dp,dp/gp,m/b,m/d,d/g}
  \draw[arc] (\from) to[bend right=18] (\to);
\end{tikzpicture}
\endgroup
}
\newcommand{\AzizSpawning}[1]{%
\begingroup
\AzizPalette{#1}
\begin{tikzpicture}[
  x=0.60cm,y=0.63cm,font=\fontsize{9}{10}\selectfont,
  dot/.style={circle,minimum size=3.6pt,inner sep=0pt},
  R/.style={draw=azizR,fill=azizR},B/.style={draw=azizB,fill=azizB},
  arc/.style={-{Stealth[length=4pt,width=4pt]},thick,black,shorten <=2pt,shorten >=2pt}
]
\node[anchor=east] at (-0.3,0) {0:};
\node[dot,R] (s0-0) at (0,0) {};
\node[anchor=east] at (-0.3,1) {1:};
\node[dot,R] (s1-0) at (0,1) {};
\node[dot,B] (s1-1) at (4,1) {};
\draw[arc] (s1-0) to[bend right=18] (s1-1);
\node[anchor=east] at (-0.3,2) {2:};
\node[dot,R] (s2-0) at (0,2) {};
\node[dot,B] (s2-2) at (2,2) {};
\node[dot,B] (s2-1) at (4,2) {};
\node[dot,R] (s2-3) at (6,2) {};
\draw[arc] (s2-0) to[bend right=18] (s2-2);
\draw[arc] (s2-0) to[bend right=18] (s2-1);
\draw[arc] (s2-1) to[bend right=18] (s2-3);
\node[anchor=east] at (-0.3,3) {3:};
\node[dot,R] (s3-0) at (0,3) {};
\node[dot,B] (s3-4) at (1,3) {};
\node[dot,B] (s3-2) at (2,3) {};
\node[dot,R] (s3-5) at (3,3) {};
\node[dot,B] (s3-1) at (4,3) {};
\node[dot,R] (s3-6) at (5,3) {};
\node[dot,R] (s3-3) at (6,3) {};
\node[dot,B] (s3-7) at (7,3) {};
\draw[arc] (s3-0) to[bend right=18] (s3-4);
\draw[arc] (s3-0) to[bend right=18] (s3-2);
\draw[arc] (s3-2) to[bend right=18] (s3-5);
\draw[arc] (s3-0) to[bend right=18] (s3-1);
\draw[arc] (s3-1) to[bend right=18] (s3-6);
\draw[arc] (s3-1) to[bend right=18] (s3-3);
\draw[arc] (s3-3) to[bend right=18] (s3-7);
\ifnum#1=0
\draw[black,thin] (0,3) -- (0.5,3.52);
\draw[black,thin] (1,3) -- (0.5,3.52);
\draw[black,thin] (2,3) -- (2.5,3.52);
\draw[black,thin] (3,3) -- (2.5,3.52);
\draw[black,thin] (4,3) -- (4.5,3.52);
\draw[black,thin] (5,3) -- (4.5,3.52);
\draw[black,thin] (6,3) -- (6.5,3.52);
\draw[black,thin] (7,3) -- (6.5,3.52);
\draw[black,thin] (0.5,3.52) -- (1.5,4.04);
\draw[black,thin] (2.5,3.52) -- (1.5,4.04);
\draw[black,thin] (4.5,3.52) -- (5.5,4.04);
\draw[black,thin] (6.5,3.52) -- (5.5,4.04);
\draw[black,thin] (1.5,4.04) -- (3.5,4.56);
\draw[black,thin] (5.5,4.04) -- (3.5,4.56);
\else
\foreach \vertex/\index/\col in {0/1/azizR,4/2/azizB,2/3/azizB,5/4/azizR,1/5/azizB,6/6/azizR,3/7/azizR,7/8/azizB}
  \node[above=3pt,text=\col] at (s3-\vertex) {$m_{\index}$};
\fi
\end{tikzpicture}
\endgroup
}

%% file: Figures/aziz-m-coloring.tex
\begingroup
\color{black}
\begin{minipage}[b]{0.47\linewidth}
\centering
\AzizSpawning{1}
\par\smallskip\small (a) Spawning iterations
\end{minipage}\hfill
\begin{minipage}[b]{0.50\linewidth}
\centering
\begin{tikzpicture}[
  x=0.82cm,y=0.75cm,font=\fontsize{9}{10}\selectfont,
  vertex/.style={circle,draw,fill=white,minimum size=14pt,inner sep=0pt,font=\fontsize{8}{9}\selectfont},
  R/.style={draw=red,text=red},B/.style={draw=blue,text=blue},
  arc/.style={-{Stealth[length=3.2pt,width=3.2pt]},thick,black,shorten <=0.5pt,shorten >=0.5pt}
]
\node[anchor=east] at (-0.45,3.70) {Depth};
\foreach \y/\depth in {3/0,2/1,1/2,0/3}
  \node[anchor=east] at (-0.45,\y) {$\depth$};
\node[vertex,R] (m1) at (2.0,3) {$m_1$};
\node[vertex,B] (m2) at (0.0,2) {$m_2$};
\node[vertex,B] (m3) at (1.3,2) {$m_3$};
\node[vertex,B] (m5) at (3.4,2) {$m_5$};
\node[vertex,R] (m4) at (1.3,1) {$m_4$};
\node[vertex,R] (m6) at (2.8,1) {$m_6$};
\node[vertex,R] (m7) at (4.0,1) {$m_7$};
\node[vertex,B] (m8) at (4.0,0) {$m_8$};
\foreach \parent/\child in {1/2,1/3,1/5,3/4,5/6,5/7,7/8}
  \draw[arc] (m\parent) -- (m\child);
\end{tikzpicture}
\par\smallskip\small (b) Depth in the rooted tree
\end{minipage}
\par\medskip
{\small Red order: {\color{red}$m_7>m_6>m_4>m_1$}\qquad
Blue order: {\color{blue}$m_8>m_5>m_3>m_2$}\par}
\endgroup

%% file: Figures/hardness-two-chains-auxiliary.tex
\begingroup
\color{black}
\begin{tikzpicture}[
thick,
scale=0.90, transform shape,
dot/.style={circle, draw, fill, minimum size=3pt, inner sep=0pt},
every node/.style={circle, draw, inner sep=2pt, minimum size=10pt},
texts/.style={fill=none, draw=none, rectangle},
arcs/.style={-{Stealth[length=6pt, inset=3pt, round, scale width=1.4]}, black, shorten <=2pt, shorten >=2pt, thick},
every fit/.style={rectangle, dashed, rounded corners=5pt, thick, draw=black, inner sep=5pt}
]
\path[use as bounding box] (-59pt,-114pt) rectangle (117pt,66pt);
\node[dot] at (50pt, 50pt) (v-star) {};
\node[texts, text=black, anchor=south west] at (v-star.north east) {$v'$};
\node[dot] at (10pt, 30pt) (b1) {};
\node[dot] at (10pt, 15pt) (b2) {};
\node[texts] at (10pt, 0pt) {$\vdots$};
\node[dot] at (10pt, -20pt) (b3) {};
\node[dot] at (10pt, -40pt) (a1) {};
\node[texts,font=\fontsize{8}{9}\selectfont,inner sep=0pt,minimum size=0pt] at (10pt,-48pt) {$C_R$};
\node[texts,font=\fontsize{8}{9}\selectfont,inner sep=0pt,minimum size=0pt] at (10pt,-62pt) {$C_B$};
\node[dot] at (10pt, -70pt) (a2) {};
\node[dot] at (90pt, -40pt) (b11) {};
\node[texts] at (90pt, -52pt) {$\vdots$};
\node[dot] at (90pt, -70pt) (b22) {};
\node[texts, anchor=east, xshift=-8pt] at (a1.west) {$a^*$};
\node[texts, anchor=west, xshift=8pt] at (b11.east) {$b^*$};
\node[fit=(b1) (b3)] (A-1) {};
\node[texts, anchor=south west, xshift = +5pt] at (a1.north east) {$A=D$};
\node[fit=(b11) (b22)] (B) {};
\node[texts, anchor=south, yshift=-18pt] (nin) at (B.south) {$\nin(v')$};
\node[texts, anchor=south east, xshift = -5pt] at (b11.north west) {$B$};
\node[fit=(a1) (a2)] (A) {};
\draw[dashed,thin] (A.west |- 0,-55pt) -- (A.east |- 0,-55pt);
\node[texts, anchor=south, yshift=-18pt] (nout) at (A.south) {$\nout(v')$};
\node[texts, anchor=south west, xshift = +5pt] at (b1.north east) {$A'$};
\draw[arcs] (v-star.west) to[out=180, in=90, looseness=1] (A-1.north);
\draw[arcs] (B.north) to[out=90, in=0, looseness=1] (v-star.east);
\draw[arcs] (v-star.north west) to[out=150, in=180, looseness=1.5] (A.west);
\draw[arcs] (B.north) to[out=90, in=0, looseness=1] (A-1.east);
\draw[arcs] (A-1.west) to[out=180, in=180, looseness=1] (A.west);
\draw[arcs, red] (a1) to[out=0, in=180, looseness=1] (b11);
\draw[arcs] (b11) -- (a2);
\draw[arcs] (b22) -- (a1);
\draw[arcs] (b22) -- (a2);
\draw[arcs] (b11) to[out=285, in=75, looseness=1] (b22);
\end{tikzpicture}
\endgroup

%% file: Figures/hardness-two-chains-final.tex
\begingroup
\color{black}
\begin{tikzpicture}[
thick,
scale=0.90, transform shape,
dot/.style={circle, draw, fill, minimum size=3pt, inner sep=0pt},
every node/.style={circle, draw, inner sep=2pt, minimum size=10pt},
texts/.style={fill=none, draw=none, rectangle},
arcs/.style={-{Stealth[length=6pt, inset=3pt, round, scale width=1.4]}, black, shorten <=2pt, shorten >=2pt, thick},
every fit/.style={rectangle, dashed, rounded corners=5pt, thick, draw=black, inner sep=5pt}
]
\path[use as bounding box] (-59pt,-114pt) rectangle (117pt,66pt);
\node[dot] at (50pt, 50pt) (v-star) {};
\node[texts, text=black, anchor=south west] at (v-star.north east) {$v'$};
\node[dot] at (10pt, 30pt) (b1) {};
\node[dot] at (10pt, 15pt) (b2) {};
\node[texts] at (10pt, 0pt) {$\vdots$};
\node[dot] at (10pt, -20pt) (b3) {};
\node[dot] at (10pt, -40pt) (a1) {};
\node[texts,font=\fontsize{8}{9}\selectfont,inner sep=0pt,minimum size=0pt] at (10pt,-48pt) {$C_R$};
\node[texts,font=\fontsize{8}{9}\selectfont,inner sep=0pt,minimum size=0pt] at (10pt,-62pt) {$C_B$};
\node[dot] at (10pt, -70pt) (a2) {};
\node[dot] at (90pt, -40pt) (b11) {};
\node[texts] at (90pt, -52pt) {$\vdots$};
\node[dot] at (90pt, -70pt) (b22) {};
\node[texts, anchor=east, xshift=-8pt] at (a1.west) {$a^*$};
\node[texts, anchor=west, xshift=8pt] at (b11.east) {$b^*$};
\node[fit=(b1) (b3)] (A-1) {};
\node[texts, anchor=south west, xshift = +5pt] at (a1.north east) {$A=D$};
\node[fit=(b11) (b22)] (B) {};

\node[texts, anchor=south east, xshift = -5pt] at (b11.north west) {$B$};
\node[fit=(a1) (a2)] (A) {};
\draw[dashed,thin] (A.west |- 0,-55pt) -- (A.east |- 0,-55pt);

\node[texts, anchor=south west, xshift = +5pt] at (b1.north east) {$A'$};
\draw[arcs] (v-star.west) to[out=180, in=90, looseness=1] (A-1.north);
\draw[arcs] (B.north) to[out=90, in=0, looseness=1] (v-star.east);
\coordinate (CR-west) at (A.west |- 0,-49pt);
\coordinate (CB-west) at (A.west |- 0,-65pt);
\draw[arcs] (v-star.north west)
    .. controls (28pt,59pt) and (-6pt,59pt) .. (-20pt,30pt)
    .. controls (-43pt,-3pt) and (-39pt,-49pt) .. (CR-west);
\draw[arcs,red] (CB-west)
    .. controls (-60pt,-65pt) and (-65pt,54pt) .. (-20pt,62pt)
    .. controls (4pt,65pt) and (50pt,66pt) .. (v-star.north);
\draw[arcs] (B.north) to[out=90, in=0, looseness=1] (A-1.east);
\draw[arcs] (A-1.west) to[out=180, in=180, looseness=1] (A.west);
\draw[arcs, red] (a1) to[out=0, in=180, looseness=1] (b11);
\draw[arcs] (b11) -- (a2);
\draw[arcs] (b22) -- (a1);
\draw[arcs] (b22) -- (a2);
\draw[arcs] (b11) to[out=285, in=75, looseness=1] (b22);
\node[texts,font=\fontsize{8}{9}\selectfont,inner sep=0pt,minimum size=0pt] at (50pt,-89pt) {$\nout(v')=A'\cup C_R$};
\node[texts,font=\fontsize{8}{9}\selectfont,inner sep=0pt,minimum size=0pt] at (50pt,-103pt) {$\nin(v')=B\cup C_B$};
\end{tikzpicture}
\endgroup

%% file: Figures/example-tight-ins.tex
    \begin{figure}[h]
\FIGURE
{\begin{minipage}{\linewidth}\centering
        \begin{tikzpicture}[
            thick,
            scale=0.90, transform shape,
            dot/.style={circle, draw, fill, minimum size=3pt, inner sep=0pt},
            every node/.style={circle, draw, inner sep=2pt, minimum size=10pt},
            texts/.style={fill=none, draw=none, rectangle},
            arcs/.style={-{Stealth[length=6pt, inset=3pt, round, scale width=1.4]}, black, shorten <=2pt, shorten >=2pt, thick},
            every fit/.style={rectangle, dashed, rounded corners=5pt, thick, draw=black, inner sep=5pt}
        ]
            \node[dot] at (50pt, 50pt) (v-star) {};
            \node[texts, anchor=south west] at (v-star.north east) {$v^*$};
            \node[dot] at (80pt, 30pt) (b11) {};                %
            \node[dot] at (80pt, 15pt) (b22) {};               %
            \node[dot] at (80pt, 0pt) (b33) {};               %
            \node[texts] at (80pt, -12pt) {$\vdots$};            %
            \node[dot] at (80pt, -30pt) (b44) {};              %
            \node[dot] at (20pt, 30pt) (b1) {};                %
            \node[dot] at (20pt, 15pt) (b2) {};                %
            \node[texts] at (20pt, 0pt) {$\vdots$};           %
            \node[dot] at (20pt, -18pt) (b3) {};                 %
            \node[texts, anchor=east, xshift=-8pt] at (b1.west) {$a'$};
            \node[fit=(b1) (b3)] (f1) {};                            %
            \node[texts, anchor=south, yshift=-18pt] (nout) at (f1.south) {$\nout(v^*)$}; 
            \node[fit=(b11) (b44)] (f2) {};                          %
            \node[texts, anchor=south, yshift=-18pt] (nout) at (f2.south) {$\nin(v^*)$}; 
            \draw[arcs] (v-star.west) to[out=180, in=90, looseness=1] (f1.north);
            \draw[arcs] (f2.north) to[out=90, in=0, looseness=1] (v-star.east);
            \draw[arcs, red] (b1) to (b11);
            \draw[arcs, red] (b1) to (b22);
            \draw[arcs, red] (b1) to (b33);
            \draw[arcs, red] (b1) to (b44);
            \draw[arcs] (b11) -- (b2);
            \draw[arcs] (b22) -- (b2);
            \draw[arcs] (b33) -- (b2);
            \draw[arcs] (b44) -- (b2);
            \draw[arcs] (b11) -- (b3);
            \draw[arcs] (b22) -- (b3);
            \draw[arcs] (b33) -- (b3);
            \draw[arcs] (b44) -- (b3);
        \end{tikzpicture}
\end{minipage}}
{An illustration of a no-instance $(D, v^*)$ where $v^*$ is a king and $\forall b\in \nin(v^*), out(b) \leq 2out(v^*)$.\label{fig:tight}}
{}
\end{figure}